\documentclass[letterpaper,twocolumn,twoside]{ieeetran}

\usepackage{amsmath,amssymb,mathtools}
\usepackage{graphicx,color,subfig}
\usepackage{cite}

\usepackage{amsthm}
\let\labelindent\relax
\usepackage{tikz,calc,bbding}
\usetikzlibrary{shapes, arrows, decorations, decorations.pathreplacing, decorations.pathmorphing, decorations.markings, fit, matrix}
\usepackage[font=small,labelfont=bf]{caption}
\usepackage{arydshln}
\usepackage[hidelinks]{hyperref}

\graphicspath{{figures/}}

\newtheorem{theorem}{Theorem}[section]
\newtheorem{proposition}[theorem]{Proposition}
\newtheorem{lemma}[theorem]{Lemma}
\newtheorem{corollary}[theorem]{Corollary}
\newtheorem{remark}[theorem]{Remark}
\newtheorem{definition}[theorem]{Definition}
\newtheorem{example}[theorem]{Example}

\newtheorem{conjecture}[theorem]{Conjecture}
\newtheorem{assumption}{Assumption}

\newcommand{\longthmtitle}[1]{\mbox{}{\bf \textit{(#1).}}}

\newcommand{\real}{\ensuremath{\mathbb{R}}}
\newcommand{\realpos}{\ensuremath{\mathbb{R}_{>0}}}
\newcommand{\realnonneg}{\ensuremath{\mathbb{R}_{\ge 0}}}
\newcommand{\realnonpos}{\ensuremath{\mathbb{R}_{\le 0}}}

\newcommand{\setdef}[2]{\{#1 \; | \; #2\}}
\newcommand{\setdefb}[2]{\big\{#1 \; | \; #2\big\}}

\newcommand{\Cc}{\mathcal{C}}

\newcommand{\Hc}{\mathcal{H}}

\newcommand{\Lc}{\mathcal{L}}

\newcommand{\Nc}{\mathcal{N}}
\newcommand{\Oc}{\mathcal{O}}
\newcommand{\Pc}{\mathcal{P}}

\newcommand\bbf{\mathbf{b}}
\newcommand\dbf{\mathbf{d}}

\newcommand\mbf{\mathbf{m}}

\newcommand\qbf{\mathbf{q}}
\newcommand\ubf{\mathbf{u}}

\newcommand\wbf{\mathbf{w}}
\newcommand\xbf{\mathbf{x}}
\newcommand\ybf{\mathbf{y}}
\newcommand\zbf{\mathbf{z}}
\newcommand\Abf{\mathbf{A}}
\newcommand\Bbf{\mathbf{B}}
\newcommand\Dbf{\mathbf{D}}
\newcommand\Ibf{\mathbf{I}}
\newcommand\Kbf{\mathbf{K}}
\newcommand\Mbf{\mathbf{M}}

\newcommand\Pbf{\mathbf{P}}
\newcommand\Qbf{\mathbf{Q}}
\newcommand\Rbf{\mathbf{R}}
\newcommand\Wbf{\mathbf{W}}

\newcommand\deltab{\boldsymbol{\delta}}
\newcommand\sigmab{{\boldsymbol{\sigma}}}
\newcommand\Sigmab{\boldsymbol{\Sigma}}
\newcommand\Gammab{\boldsymbol{\Gamma}}
\newcommand\xib{\boldsymbol{\xi}}
\newcommand\nub{\boldsymbol{\nu}}

\newcommand\Pib{\boldsymbol{\Pi}}
\newcommand\ellb{{\boldsymbol{\ell}}}

\newcommand{\ones}{\mathbf{1}}
\newcommand{\zeros}{\mathbf{0}}

\newcommand{\diag}{\text{diag}}
\newcommand{\until}[1]{\{1,\dots,#1\}}

\renewcommand\det{\text{det}}
\newcommand\izon{\in \{0, 1\}^n}
\newcommand\range{\text{range}}
\newcommand\ssm{{\raisebox{2pt}{\scriptsize
      $\mathfrak{0}$}}}%
\newcommand\ssp{{\raisebox{2pt}{\scriptsize
      $\mathfrak{1}$}}}%
\newcommand\s{{\rm s}}
\newcommand\zls{\{0, \ell, \s\}}

\usepackage{enumitem}
\renewcommand{\theenumi}{(\roman{enumi})}

\renewcommand{\footnoterule}{%
  \hspace{3pt} \hrule width 0.4\textwidth height 0.5pt
  \kern 2pt
}

\newcommand{\oprocendsymbol}{\hbox{$\square$}}
\newcommand{\oprocend}{\relax\ifmmode\else\unskip\hfill\fi\oprocendsymbol}

\renewcommand{\baselinestretch}{0.99}
\allowdisplaybreaks

\tikzset{->-/.style={decoration={
  markings,
  mark=at position #1 with {\arrow{>}}},postaction={decorate}}}
  
\title{\Large \bf Hierarchical Selective Recruitment in
  Linear-Threshold Brain Networks
  \\
  Part I: Single-Layer Dynamics and Selective Inhibition\thanks{A
    preliminary version appeared as~\cite{EN-JC:18-acc} at the
    American Control Conference.}}

\author{Erfan Nozari \quad Jorge Cort\'es\thanks{Erfan Nozari is with
    the Department of Electrical and Systems Engineering, University
    of Pennsylvania, Philadelphia, PA 19104,
    enozari@seas.upenn.edu. Jorge Cort\'es is with the Department of
    Mechanical and Aerospace Engineering, UC San Diego, La Jolla, CA
    92093, cortes@ucsd.edu.}}

\begin{document}

\maketitle
\thispagestyle{empty}
\pagestyle{empty}

\begin{abstract}
  Goal-driven selective attention (GDSA) refers to the brain's
  function of prioritizing the activity of a task-relevant subset of
  its overall network to efficiently process relevant information
  while inhibiting the effects of distractions. Despite decades of
  research in neuroscience, a comprehensive understanding of GDSA is
  still lacking.  We propose a novel framework using concepts and
  tools from control theory as well as insights and structures from
  neuroscience. Central to this framework is an information-processing
  hierarchy with two main components: selective inhibition of
  task-irrelevant activity and top-down recruitment of task-relevant
  activity.  We analyze the internal dynamics of each layer of the
  hierarchy described as a network with linear-threshold dynamics and
  derive conditions on its structure to guarantee existence and
  uniqueness of equilibria, asymptotic stability, and boundedness of
  trajectories.  We also provide mechanisms that enforce selective
  inhibition using the biologically-inspired schemes of feedforward
  and feedback inhibition. Despite their differences, both lead to the
  same conclusion: the intrinsic dynamical properties of the
  (not-inhibited) task-relevant subnetworks are the sole determiner of
  the dynamical properties that are achievable under selective
  inhibition.
\end{abstract}

\section{Introduction}\label{sec:intro}

The human brain is constantly under the influx of sensory inputs and
is responsible for integrating and interpreting them to generate
appropriate decisions and actions.  This influx contains not only the
pieces of information relevant to the present task(s), but also a
myriad of distractions.  Goal-driven selective attention (GDSA) refers
to the active selective processing of a subset of information influx
while suppressing the effects of others, and is vital for the proper
function of the brain.%
\footnote{Note the distinction with stimulus-driven selective
  attention (the reactive shift of focus based on saliency of stimuli)
  which is not the focus here.} 
Examples range from selective audition in a crowded place to selective
vision in cluttered environments to selective taste/smell in food.  As
a result, a long standing question in neuroscience involves
understanding the brain's complex mechanisms underlying selective
attention~\cite{JS:90,ECC:53,RD-JD:95,LI-CK:01,NL-AH-JWD-EV:04,AG-ACN:12}.

A central element in addressing this question is the role played by
the hierarchical organization of the brain~\cite{JTS-SK:14}. Broadly,
this organization places primary sensory and motor areas at the bottom
and integrative association areas (prefrontal cortex in particular) at
the top. Accordingly, sensory information is processed while flowing
up the hierarchy, where decisions are eventually made and transmitted
back down the hierarchy to generate motor actions.%
\footnote{Note that the role of memory (being distributed across the
  brain) is implicit in this simplified stimulus-response
  description. Indeed, many sensory inputs only form memories (without
  motor response) or many motor actions result chiefly from memory
  (without sensory stimulation). The hierarchical aspect is
  nevertheless present.}  The top-down direction is also responsible
for GDSA, where the higher-order areas differentially ``modulate'' the
activity of the lower-level areas such that only relevant information
is further processed. This phenomenon constitutes the basis for GDSA
and has been the subject of extensive experimental research in
neuroscience, see
e.g.,~\cite{DEB:58,AMT:69,JM-RD:85,BCM:93,RD-JD:95,SK-PD-RD-LGU:98,MAP-GMD-SK:04,NL:05,JJF-ACS:11,HP:16,MG-KH-EN:16}. However,
a complete understanding of how, when (how quick), or where (within
the hierarchy) it occurs is still lacking. In particular, the
relationship between GDSA and the dynamics of the involved neuronal
networks is poorly understood. Our goal is to address this gap from a
model-based perspective, resorting to control-theoretic tools to
explain various aspects of GDSA in terms of the synaptic network
structure and the dynamics that emerge from~it.

In this work, we propose a theoretical framework, termed
\emph{hierarchical selective recruitment (HSR)}, to explain the
network dynamics underlying GDSA.  This framework consists of a novel
hierarchical model of brain organization (though composed of
well-established sub-models at each layer), a set of analytical
results regarding the multi-timescale dynamics of this model, and a
careful translation between the properties of these dynamics and well
known experimental observations about GDSA.  The starting point in the
development of HSR is the observation that different stimuli, in
particular the task-relevant and task-irrelevant ones, are processed
by different populations of neurons (see,
e.g.,~\cite{JM-RD:85,BCM:93,RD-JD:95,SK-PD-RD-LGU:98,LI-CK:01,MAP-GMD-SK:04,AG-ACN:12,MG-KH-EN:16}). With
each neuronal population represented by a node in the overall neuronal
network of networks and based on extensive experimental research (see
below), HSR primarily relies on the \emph{selective inhibition} of the
task-irrelevant nodes and the \emph{top-down recruitment} of the
task-relevant nodes of each layer by the layer immediately above.
This paper analyzes the dynamics of individual layers as well as the
mechanisms for selective inhibition in a bilayer network. These
results set the basis for the study of the mechanisms for top-down
recruitment in multilayer networks in our accompanying
work~\cite{EN-JC:21-tacII}.

\subsubsection*{Literature Review}

In this work we use dynamical networks with linear-threshold
nonlinearities (the unbounded version also called rectified linear
units, ReLU, in machine learning) to model the activity of neuronal
populations.  Linear-threshold models allow for a unique combination
between the tractability of linear systems and the dynamical
versatility of nonlinear systems, and thus have been widely used in
computational neuroscience.  They were first proposed as a model for
the lateral eye of the horseshoe crab in~\cite{FR-HKH:74} and their
dynamical behavior has been studied at least as early
as~\cite{KPH:73}. A detailed stability analysis of symmetric
(undirected) linear-threshold networks has been carried out in
continuous~\cite{RHRH-HSS-JJS:03} and discrete~\cite{ZY-LZ-JY-KKT:09}
time: however, this has limited relevance for biological neuronal
networks, which are fundamentally asymmetric (due to the presence of
excitatory and inhibitory neurons). Regarding asymmetric networks, it
was claimed (without proper proof) in~\cite{KPH-DK:87} that the
identity minus the matrix of synaptic connectivities being a P-matrix
is necessary and sufficient for the existence and uniqueness of
equilibria (EUE), the negative of this matrix being totally Hurwitz is
necessary and sufficient for local asymptotic stability, and the
matrix of synaptic connectivities being absolutely Schur stable is
sufficient for global asymptotic stability. In addition to lacking
proper proof, these results were limited to fully-inhibitory
networks. The latter assertion was later proved rigorously
in~\cite{JF-KPH:96} for arbitrary networks, while we prove the first
two (except on certain sets of measure zero) here.  Around the same
time, \cite{MF-AT:95} considered the more general class of
monotonically non-decreasing activation functions and proved the
sufficiency of identity minus the matrix of synaptic connectivities
being a P-matrix for the uniqueness of equilibria (being only one of
the four implications we prove here) and the sufficiency of the same
matrix being Lyapunov diagonally stable for global asymptotic
stability (which we relax here by allowing for arbitrary quadratic
Lyapunov functions). This work was later generalized to discontinuous
neural networks (though not applicable to our model here)
in~\cite{MF-PN:03}. Also related is the work~\cite{IWS-ANW:69} showing
the necessity and sufficiency of identity minus the matrix of synaptic
connectivities being a P$_0$-matrix for EUE of similar systems but
with \emph{strictly} monotonically increasing activation
functions. The work~\cite{HZ-ZW-DL:14} provides a comprehensive review
of stability analysis of a range of continuous-time recurrent neural
networks, including the linear-threshold model.

Lyapunov-based methods have also been used in a number of later
studies for discrete-time linear-threshold
networks~\cite{ZY-KKT:04,WZ-JMZ:14,TS-IRP:16}, but the extension of
these results to continuous-time dynamics is unclear. In fact, the use
of Lyapunov-based techniques in continuous-time networks has remained
limited to planar dynamics~\cite{KCT-HT-WZ:05} and restrictive
conditions for boundedness of
trajectories~\cite{HW-WB-HR:01,KCT-HT-WZ:05}. Recently,~\cite{KM-AD-VI-CC:16}
presents interesting properties of competitive (i.e., fully
inhibitory) linear-threshold networks, particularly regarding the
emergence of oscillations.  However, the majority of neurons in
biological neuronal networks are excitatory, making the implications
of these results limited. Moreover, all the preceding works are
limited to networks with constant exogenous inputs whereas
time-varying inputs are essential for modeling inter-layer connections
in HSR.

A critical property of linear-threshold networks is that their
nonlinearity, while enriching their behavior beyond that of linear
systems, is piecewise linear. Accordingly, almost all the theoretical
analysis of these networks builds upon the formulation of them as
switched affine systems. There exists a vast literature on the
analysis of general switched linear/affine systems, see,
e.g.,~\cite{HL-PJA:09,DL:03,MKJJ:03}. Nevertheless, we have found that
the conditions obtained by applying these results to linear-threshold
dynamics are more conservative than the ones we obtain using direct
analysis of the system dynamics. This is mainly due to the fact that
such results, by the essence of their generality, are oblivious to the
particular structure of linear-threshold dynamics that can be
leveraged in direct analysis.

Selective inhibition has been the subject of extensive research in
neuroscience. A number of early
studies~\cite{RD-JD:95,JM-RD:85,BCM:93} provided evidence for a
mechanism of selective visual attention based on a biased competition
between the subnetwork of task-relevant nodes and the subnetwork of
task-irrelevant ones. In this model, nodes belonging to these
subnetworks compete at each layer by mutually suppressing the activity
of each other, and this competition is biased towards task-relevant
nodes by the layer immediately above. Later
studies~\cite{SK-PD-RD-LGU:98,MAP-GMD-SK:04} further supported this
theory using functional magnetic resonance imaging (fMRI) and
showed~\cite{KNS-MVP-SK:12}, in particular, the suppression of
activity of task-irrelevant nodes as a result of GDSA. This
suppression of activity is further shown to occur in multiple layers
along the hierarchy~\cite{SK-PD-MAP-MIE-RD-LGU:01}, grow with
increasing attention~\cite{GR-CDF-NL:97,DHO-MMF-MAP-SK:02}, and be
inversely related to the power of the task-irrelevant nodes' state
trajectories in the alpha frequency band ($\sim
8$-$14^\text{Hz}$)~\cite{JJF-ACS:11}. 

\subsubsection*{Statement of Contributions}
The contributions are twofold.  First, we analyze the internal
dynamics of a single-layer linear-threshold network as a basis for our
study of hierarchical structures.  Our results here provide a
comprehensive characterization of the dynamical properties of
linear-threshold networks. Specifically, we show that existence and
uniqueness of equilibria, asymptotic stability, and boundedness of
trajectories can be characterized using simple algebraic conditions on
the network structure in terms of the class of P-matrices (matrices
with positive principal minors), totally-Hurwitz matrices (those with
Hurwitz principal submatrices, shown to be a sub-class of P-matrices),
and Schur-stable matrices, respectively. In addition to forming the
basis of HSR, these results solve some long-standing open problems in
the characterization of linear-threshold
networks~\cite{KPH:73,KPH-DK:87,JF-KPH:96,HW-WB-HR:01,KCT-HT-WZ:05,KM-AD-VI-CC:16}
and are of independent interest.  Our analysis covers both the class
of unbounded (a.k.a. ReLU) as well as bounded linear-threshold
networks, where the latter is a piecewise-affine approximation of
sigmoidal neural networks, for which limited analytical results are
available. Our second contribution pertains the problem of selective
inhibition in a bilayer network.  Motivated by the mechanisms of
inhibition in the brain, we study feedforward and feedback mechanisms.
We provide necessary and sufficient conditions on the network
structure that guarantee selective inhibition of task-irrelevant nodes
at the lower-level while simultaneously guaranteeing various dynamical
properties of the resulting (partly inhibited, partly active)
subnetwork, including existence and uniqueness of equilibria and
asymptotic stability. Interestingly, under both mechanisms, these
conditions require that the (not-inhibited) task-relevant part of the
lower-level subnetwork intrinsically satisfies the same desired
dynamical properties. This is particularly important for selective
inhibition as asymptotic stability underlies it.  The results unveil
the important role of task-relevant nodes in constraining the
dynamical properties achievable under selective inhibition and have
implications for the number and centrality of nodes that need to be
inhibited for an unstable-in-isolation subnetwork to gain stability
through selective inhibition.  For subnetworks that are not stable as
a whole, these results provide conditions on the
task-relevant/irrelevant partitioning of the nodes that allow for
stabilization using inhibitory control.

\section{Preliminaries}\label{sec:prelims}

We introduce notational conventions and basic concepts on matrix
analysis and modeling of biological neuronal networks.

\subsection*{Notation}
Throughout the paper, we employ the following notation. We use
$\real$, $\realnonneg$, and $\realnonpos$ to denote the set of reals,
nonnegative reals, and nonpositive reals, respectively.  We use
bold-faced letters for vectors and matrices. $\ones_n$, $\zeros_n$,
$\ellb_n$, $\zeros_{p \times n}$, and $\Ibf_n$ stand for the
$n$-vector of all ones, the $n$-vector of all zeros, the $n$-vector of
all $\ell$'s, the $p$-by-$n$ zero matrix, and the identity $n$-by-$n$
matrix (we omit the subscripts when clear from the context).  Given a
vector $\xbf \in \real^n$, $x_i$ and $(\xbf)_i$ refer to its $i$th
component.  Given $\Abf \in \real^{p \times n}$, $a_{i j}$ refers to
the $(i, j)$th entry. For block-partitioned $\xbf$ and $\Abf$,
$\xbf_i$ and $\Abf_{ij}$ refer to the $i$th block of $\xbf$ and $(i,
j)$th block of $\Abf$, respectively.  In block representation of
matrices, $\star$ denotes arbitrary blocks whose value is immaterial
to the discussion. For $\Abf \in \real^{p \times n}$, $\range(\Abf)$
denotes the subspace of $\real^p$ spanned by the columns of $\Abf$.
If $\xbf$ and $\ybf$ are vectors, $\xbf \le \ybf$ denotes $x_i \le
y_i$ for all~$i$.  For symmetric $\Pbf \in \real^{n \times n}$, $\Pbf
> \zeros$ ($\Pbf < \zeros$) denotes that $\Pbf$ is positive (negative)
definite.  Given $\Abf \in \real^{n \times n}$, its element-wise
absolute value, determinant, spectral radius, and induced $2$-norm are
denoted by $|\Abf|$, $\det(\Abf)$, $\rho(\Abf)$, and $\|\Abf\|$,
respectively. Similarly, for $\xbf \in \real^n$, $\|\xbf\|$ is its
$2$-norm.  Likewise, for two matrices $\Abf$ and $\Bbf$, $\diag(\Abf,
\Bbf)$ denotes the block-diagonal matrix with $\Abf$ and $\Bbf$ on its
diagonal. Given a subspace $W$ of $\real^n$, $W^\perp$ denotes the
orthogonal complement of~$W$ in~$\real^n$. For $x \in \real$ and $m
\in \realpos \cup \{\infty\}$, %
$[x]_0^m = \min\{\max\{x, 0\}, m\}$, which is the projection of $x$
onto $[0, m]$. When $\xbf \in \real^n$ and $\mbf \in \realpos^n \cup
\{\infty\}^n$, we similarly define $[\xbf]_\zeros^\mbf =
[[x_1]_0^{m_1} \quad \cdots \quad [x_n]_0^{m_n}]^T$.  All
measure-theoretic statements are meant in the Lebesgue sense.

\subsection*{Matrix Analysis}

We here define and characterize several matrix classes of
interest that play a key role in the forthcoming discussion.

\begin{definition}\longthmtitle{Matrix
    classes}\label{def:matrix-classes}
  A matrix $\Abf \in \real^{n \times n}$ is
  \begin{enumerate}
  \item \emph{absolutely Schur stable} if $\rho(|\Abf|) < 1$;
  \item \emph{totally $\Lc$-stable}, denoted $\Abf \in \Lc$, if there
    exists $\Pbf = \Pbf^T > \zeros$ such that $(-\Ibf + \Abf^T
    \Sigmab) \Pbf + \Pbf (-\Ibf + \Sigmab \Abf) < \zeros$ for $\Sigmab
    = \diag(\sigmab)$ and all $ \sigmab \in \{0, 1\}^n$;
  \item \emph{totally Hurwitz}, denoted $\Abf \in \Hc$, if all the
    principal submatrices of $\Abf$ are Hurwitz;
  \item \emph{a P-matrix}, denoted $\Abf \in \Pc$, if all the
    principal minors of $\Abf$ are positive.
  \end{enumerate}
\end{definition}

In working with P-matrices, the principal pivot transform of a matrix
plays an important role. 
Given
\begin{align*}
  \Abf = 
  \begin{bmatrix}
    \Abf_{11} & \Abf_{12} \\ \Abf_{21} &
    \Abf_{22}     
  \end{bmatrix},
\end{align*}
with nonsingular $\Abf_{22}$, its principal pivot transform is the matrix
\begin{align*}
  \pi(\Abf) \triangleq 
  \begin{bmatrix}
    \Abf_{11} - \Abf_{12} \Abf_{22}^{-1} \Abf_{21} & \Abf_{12}
    \Abf_{22}^{-1} \\ -\Abf_{22}^{-1} \Abf_{21} & \Abf_{22}^{-1}
  \end{bmatrix}.
\end{align*}
Note that $\pi(\pi(\Abf)) = \Abf$.
The next result formalizes
several equivalent characterizations of P-matrices.

\begin{lemma}\longthmtitle{Properties of P-matrices~\cite{MF-VP:62,
      OS-MT:08}}\label{lem:p-mat}
  $\Abf \in \real^{n \times n}$ is a P-matrix if and only if any of the following holds:
  \begin{enumerate}
  \item $\Abf^{-1}$ is a P-matrix;
  \item all real eigenvalues of all the principal submatrices of $\Abf$ are
    positive;
  \item for any $\xbf \in \real^n \!\setminus\! \{\zeros\}$ there is $k$ such
    that $x_k (\Abf \xbf)_k \!>\! 0$;
  \item the principal pivot transform of $\Abf$ is a P-matrix. 
  \end{enumerate}
\end{lemma}

The matrix classes in Definition~\ref{def:matrix-classes} have
important inclusion relationships, as shown next.

\begin{lemma}\longthmtitle{Inclusions among matrix
    classes}\label{lem:inc-mat-class}
  For $\Abf, \Wbf \in \real^{n \times n}$, we have
  \begin{enumerate}
  \item $\rho(|\Wbf|) < 1 \Rightarrow -\Ibf + \Wbf \in \Hc$;
  \item $\|\Wbf\| < 1 \Rightarrow \Wbf \in \Lc$;
  \item $\Wbf \in \Lc \Rightarrow -\Ibf + \Wbf \in \Hc$;
  \item $\Abf \in \Hc \Rightarrow -\Abf \in \Pc$.
  \end{enumerate} 
\end{lemma}
\begin{proof}
  \emph{(i)}. From~\cite[Fact 4.11.19]{DSB:09}, we have that
  $\rho(|\Wbf_\sigmab|) < 1$ for any principal submatrix
  $\Wbf_\sigmab$ of $\Wbf$, which implies $\rho(\Wbf_\sigmab) < 1$
  by~\cite[Fact 4.11.17]{DSB:09}, implying the result.

  \emph{(ii)} It is straightforward to check that $\Pbf = \Ibf_n$
  satisfies $(-\Ibf + \Wbf^T \Sigmab) \Pbf + \Pbf (-\Ibf + \Sigmab
  \Wbf) < \zeros$ for all $\sigmab \in \{0, 1\}^n$.

  \emph{(iii)} Pick an arbitrary $\sigmab \in \{0, 1\}^n$ and let the
  permutation $\Pib \in \real^{n \times n}$ be such that
  $\Pib \Sigmab \Wbf \Pib^T = \begin{bmatrix} \zeros & \zeros \\
    \hat \Wbf_{21} & \hat \Wbf_{22} \end{bmatrix}$, where $\hat
  \Wbf_{22}$ is the principal submatrix of $\Wbf$ corresponding to
  $\sigmab$. Then
    \begin{align*}
      \Pbf(-\Ibf + \Sigmab \Wbf) &= \Pbf \Pib^T \begin{bmatrix} -\Ibf & \zeros
          \\ \hat \Wbf_{21} & -\Ibf + \hat \Wbf_{22} \end{bmatrix} \Pib
      \\
      &= \Pib^T \Big(\underbrace{\Pib \Pbf \Pib^T}_{\hat \Pbf}
        \begin{bmatrix} -\Ibf & \zeros \\ \hat \Wbf_{21} & -\Ibf + \hat
            \Wbf_{22} \end{bmatrix} \Big) \Pib
      \\
      &= \Pib^T \begin{bmatrix} \star & \star \\ \star & \hat
          \Pbf_{22} (-\Ibf + \hat \Wbf_{22}) \end{bmatrix} \Pib,
    \end{align*}
    where $\hat \Pbf = \begin{bmatrix} \hat \Pbf_{11} & \hat \Pbf_{12}
        \\ \hat \Pbf_{21} & \hat \Pbf_{22} \end{bmatrix} = \hat \Pbf^T >
    \zeros$. Thus, by assumption,
    \begin{align*}
      &\Pib^T \begin{bmatrix} \star & \star \\ \star & (-\Ibf +
          \hat \Wbf_{22}^T) \hat \Pbf_{22} + \hat \Pbf_{22} (-\Ibf + \hat
          \Wbf_{22}) \end{bmatrix} \Pib < \zeros
      \\
      &\Rightarrow \begin{bmatrix} \star & \star \\ \star &
          (-\Ibf + \hat \Wbf_{22}^T) \hat \Pbf_{22} + \hat \Pbf_{22} (-\Ibf + \hat
          \Wbf_{22}) \end{bmatrix} < \zeros
      \\
      &\Rightarrow (-\Ibf + \hat \Wbf_{22}^T) \hat \Pbf_{22} + \hat \Pbf_{22} (-\Ibf
      + \hat \Wbf_{22}) < \zeros,
    \end{align*}
    proving that $-\Ibf + \hat \Wbf_{22}$ is Hurwitz. Since $\sigmab$ is
    arbitrary, $-\Ibf + \Wbf$ is totally Hurwitz.
    
    \emph{(iv)} The result follows from Lemma~\ref{lem:p-mat}(ii).
\end{proof}

\begin{remark}\longthmtitle{Counterexamples for converses of
    Lemma~\ref{lem:inc-mat-class}}\label{rem:ce}
  {\rm The converse of the implications in
    Lemma~\ref{lem:inc-mat-class} do not hold, as shown in the
    following.  First, for a general matrix $\Wbf$, neither of
    $\rho(|\Wbf|)$ and $\|\Wbf\|$ is bounded by the other. The former
    is larger for $[8, 3; 2, -1]$, e.g., while the latter is larger
    for $[0, 0; 1, 0]$. However, if $\Wbf$ satisfies the \emph{Dale's
      law} (as many biological neuronal networks do), i.e., each
    column is either nonnegative or nonpositive, then $\Wbf = |\Wbf|
    \Dbf$ where $\Dbf$ is a diagonal matrix such that $|\Dbf| =
    \Ibf$. Then, $\|\Wbf\| = \| |\Wbf| \| \ge \rho(|\Wbf|),$ showing
    that, in this case, $\rho(|\Wbf|) < 1$ is a less restrictive
    condition. Further, $-\Ibf + \Wbf \in \Hc \not\Rightarrow
    \rho(|\Wbf|) < 1$ as seen, e.g., from $\Wbf = -2\Ibf$. The same
    example shows $\Wbf \in \Lc \not\Rightarrow \|\Wbf\| <
    1$. Likewise, $-\Ibf + \Wbf \in \Hc \not\Rightarrow \Wbf \in \Lc$,
    for which $[0.5, -3; 4, -1]$ serves as a counter example (note
    that $\Wbf \in \Lc$ is an LMI feasibility problem that can be
    checked using standard solvers such as MATLAB \texttt{feasp}
    function). Finally, $\Abf = [-1, -5, 0; 0, -1, -6; -1, 0, -1]$
    ensures the converse of Lemma~\ref{lem:inc-mat-class}(iv) does not
    hold either.  \oprocend}
\end{remark}
 
Figure~\ref{fig:mat-class} shows a summary of
Lemma~\ref{lem:inc-mat-class} and Remark~\ref{rem:ce}.

\begin{figure}
  \centering
  \includegraphics[width=.875\linewidth]{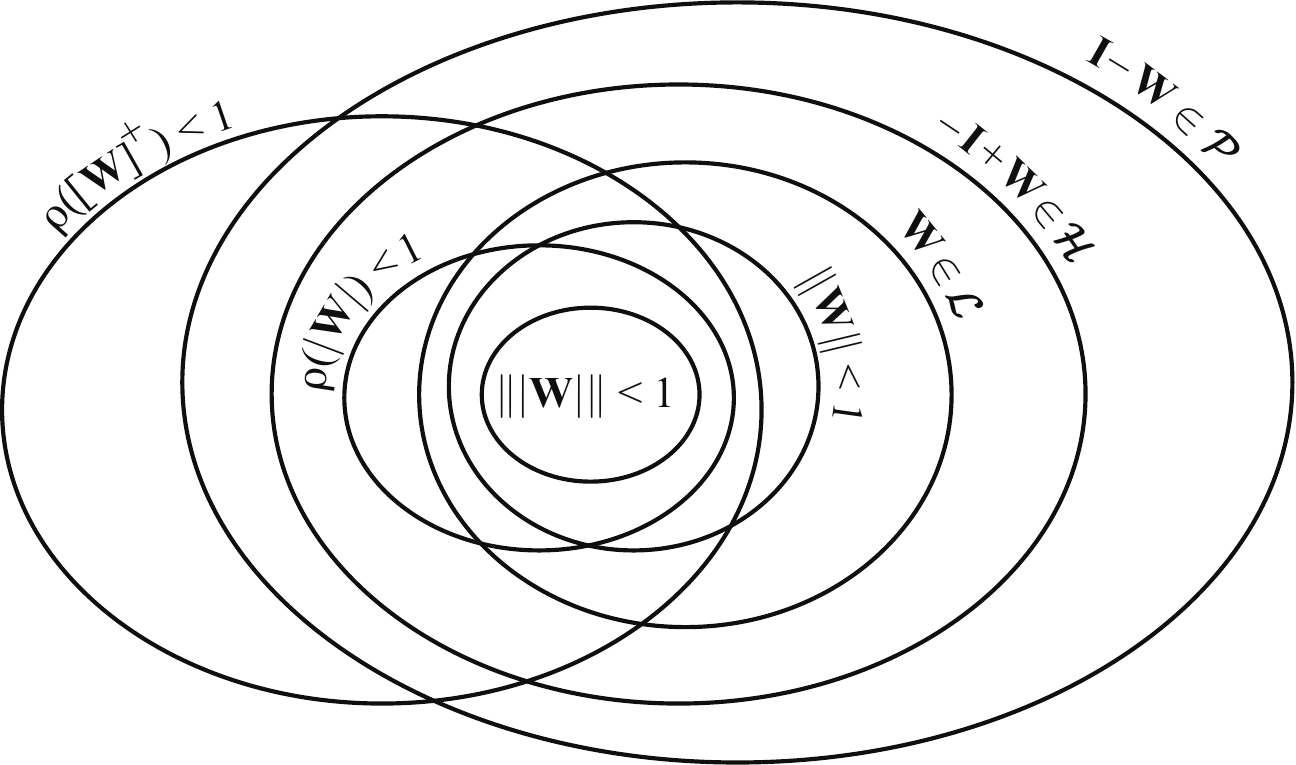}
  \caption{Inclusion relationships between the matrix classes
    introduced in Definition~\ref{def:matrix-classes} (cf.
    Lemma~\ref{lem:inc-mat-class}).}
  \label{fig:mat-class}
\vspace*{-1ex}
\end{figure}

\subsection*{Dynamical Rate Models of Brain
  Networks}

Here we briefly review, following~\cite[\S 7]{PD-LFA:01}, the
construction of the linear-threshold network model used throughout the
paper.  In a lumped model, neurons are the smallest unit of neuronal
circuits and the (directional) transmission of activity from one
neuron to another takes place at a \emph{synapse}, thus the terms
\emph{pre-synaptic} and \emph{post-synaptic} for the two neurons,
respectively.  Both the input and output signals mainly consist of a
sequence of spikes (action-potentials, Figure~\ref{fig:spike-rate} top
panel) which are modeled as impulse trains of the form
\begin{align*}
  \rho(t) = \sum_k \delta(t - t_k),
\end{align*}
where $\delta(\cdot)$ denotes the Dirac delta function. In many brain
areas, the exact timing $\{t_k\}$ of $\rho(t)$ seems highly random
while the firing rate (number of spikes per second,
Figure~\ref{fig:spike-rate} bottom panel) shows greater trial-to-trial
reproducibility. Therefore, a standard approximation is to model
$\rho(t)$ as the sample path of an inhomogeneous Poisson point process
with rate, say, $x(t)$.

Now, consider a pair of pre- and post-synaptic neurons with rates
$x_\text{pre}(t)$ and $x_\text{post}(t)$, respectively.  As a result
of $x_\text{pre}(t)$, an electrical current $I_\text{post}(t)$ flows
in the post-synaptic neuron.  Assuming fast synaptic dynamics,
$I_\text{post}(t) \propto x_\text{pre}(t)$. Let $w_\text{post,pre}$ be
the proportionality constant, so $I_\text{post}(t) = w_\text{post,pre}
x_\text{pre}(t)$. The \emph{pre-synaptic neuron} is called excitatory
if $w_\text{post,pre} > 0$ and inhibitory if $w_\text{post,pre} <
0$. In other words, excitatory neurons increase the activity of their
out-neighbors while inhibitory neurons decrease it.%
\footnote{While many brain networks, such as mammalian cortical
  networks, satisfy the Dale's law (cf. Remark~\ref{rem:ce}), all of
  our results in this work are applicable to arbitrary synaptic sign
  patterns.}  If the post-synaptic neuron receives input from multiple
neurons, $I_\text{post}(t)$ follows a superposition law,
\begin{align}\label{eq:super-pos}
  I_\text{post}(t) = \sum_j w_{\text{post}, j} x_j(t),
\end{align} 
where the sum is taken over its in-neighbors.

If $I_\text{post}$ is constant, the post-synaptic rate approximately
follows $x_\text{post} = F(I_\text{post})$, where $F$ is a nonlinear
``response function''. Among the two widely used response functions, %
sigmoidal and linear-threshold, we use the latter for its analytical
tractability: $F(\cdot) = [\, \cdot \, ]_0^{m_\text{post}}$. Finally,
if $I_\text{post}(t)$ is time-varying, $x_\text{post}(t)$ ``lags''
$F(I_\text{post}(t))$ with a time constant~$\tau$, i.e.,
\begin{align}\label{eq:res-func}
  \tau \dot x_\text{post}(t) =
    -x_\text{post}(t) + [I_\text{post}(t)]_0^{m_\text{post}}.
\end{align}
Equations~\eqref{eq:super-pos}-\eqref{eq:res-func} are the basis for
our network model described next.

\begin{figure}
  \centering
  \includegraphics[width=.9\linewidth]{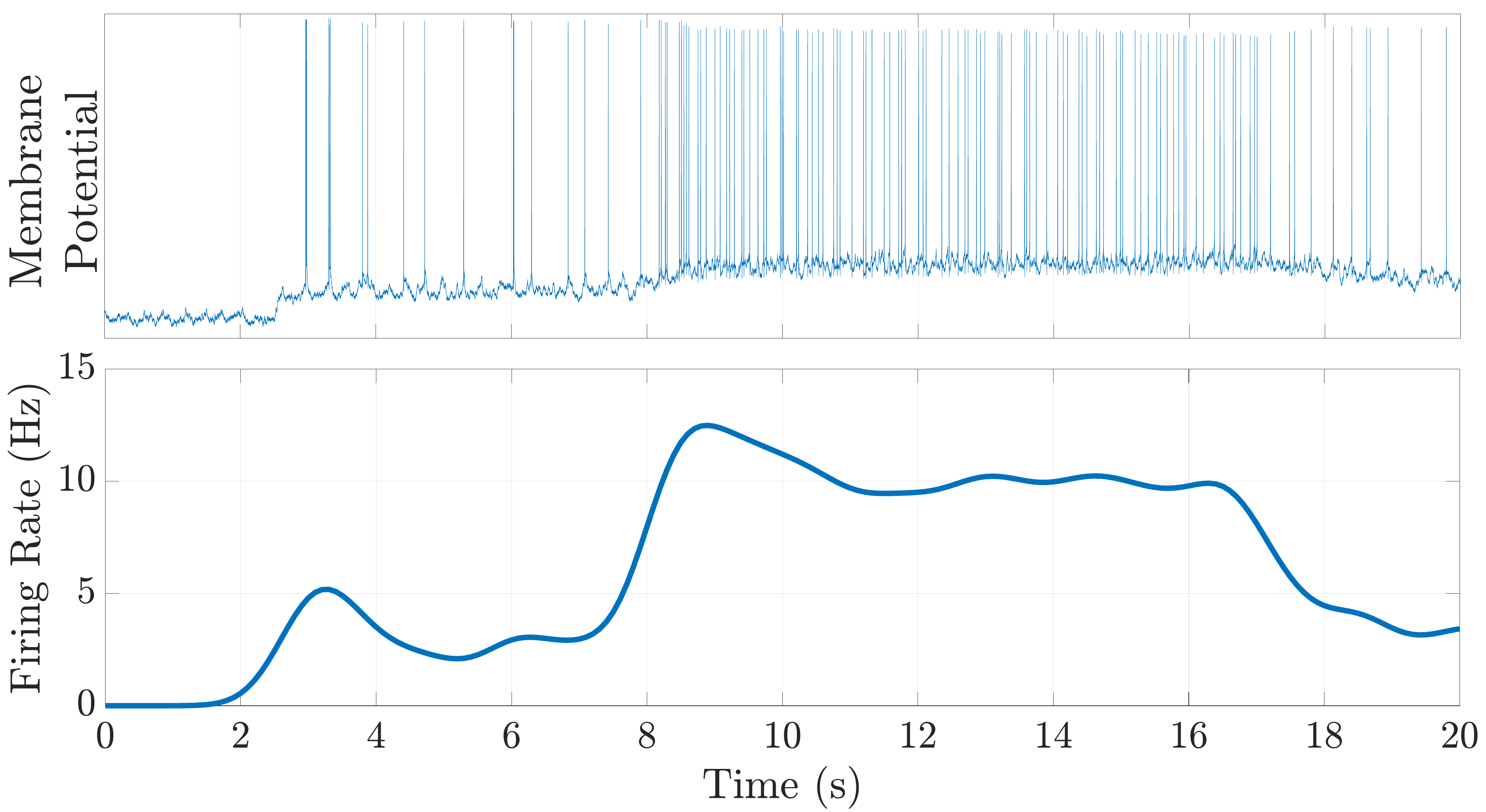}
\caption{A sample intracellular recording illustrating the spike train
  used for neuronal communication (top panel, measured
  intracellularly~\cite{DAH-ZB-JC-MAM-KDH-GB:00,DAH-KDH-ZB-JC-AM-HH-AS-GB:09-crcns})
  and the corresponding (estimate of) firing rate (bottom panel,
  estimated by binning spikes in $100^\text{ms}$ bins and smoothing
  with Gaussian window with $500^\text{ms}$ standard deviation).}
\label{fig:spike-rate} 
\vspace*{-1ex}
\end{figure}

\section{Problem Formulation}\label{sec:prob-form}

Consider a network of neurons evolving according
to~\eqref{eq:super-pos}-\eqref{eq:res-func}.  Since the number of
neurons in a brain region is very large, it is common to consider a
\emph{population of neurons} with similar activation patterns as a
single \emph{node} with the average firing rate of its neurons.  This
convention also has the advantage of getting more consistent rates, as
the firing pattern of individual neurons may be sparse.%
\footnote{Our discussion is nevertheless valid irrespective of whether
  network nodes represent individual neurons or groups of them.}
Combining the nodal rates in a vector $\xbf \in \real^n$ and synaptic
weights in a matrix $\Wbf \in \real^{n \times n}$, we obtain,
according to~\eqref{eq:super-pos}-\eqref{eq:res-func}, the
\emph{linear-threshold network dynamics}
\begin{align}\label{eq:dyn}
  \!\!\!\!\!\!\tau \dot \xbf(t) = -\xbf(t) + [\Wbf \xbf(t) +
  \dbf(t)]_\zeros^\mbf, \quad &\zeros \le \xbf(0) \le \mbf,
  \\
  \notag &\mbf \in \realpos^n \cup \{\infty\}^n.
\end{align}
The term $\dbf(t) \in \real^n$ captures the \emph{external inputs} to
the network, including un-modeled background activity and possibly
nonzero thresholds (i.e., if a node $i$ becomes active when $(\Wbf
\xbf + \dbf)_i > \vartheta_i$ for some threshold $\vartheta_i \neq
0$).

The vector of state upper bounds $\mbf$ can be finite ($\mbf \in
\realpos^n$) or infinite ($\mbf = \infty\ones_n$). Even though all
biological neurons eventually saturate for high input values, whether
finite or infinite $\mbf$ gives a more realistic/appropriate model can
vary from brain region to brain region depending on whether typical
(in vivo) values of $\xbf$ reach (near) saturation. Historically, the
unbounded case has in fact been used and studied more extensively,
both in computational neuroscience and machine learning. See, e.g.,
\cite{YA-DBR-KDM:13} and the references therein for evidence in favor
of unbounded activation functions. Surprisingly, however, the
analytical properties of the two cases are very similar, as we will
see throughout this work. Further, note that the right-hand side
of~\eqref{eq:dyn} is globally Lipschitz-continuous (though not smooth)
and therefore a unique continuously differentiable solution exists for
all $t \ge 0$~\cite[Thm 3.2]{HKK:02}. Moreover, it is straightforward
to show that if $\zeros \le \xbf(0) \le \mbf$ then $\zeros \le \xbf(t)
\le \mbf$ for all $t \ge 0$.

Consistent with the vision for hierarchical selective recruitment
(HSR) outlined in Section~\ref{sec:intro}, we consider a hierarchy of
linear-threshold networks of the form~\eqref{eq:dyn}, as depicted in
Figure~\ref{fig:multi}.  For each layer $i$, we use $\Nc_i$,
$\Nc_i^\ssp$, and $\Nc_i^\ssm, i \in \until{N}$ to denote the
corresponding subnetwork and its task-relevant and task-irrelevant
sub-subnetworks, respectively.

\begin{figure}[tb]
  \begin{center} 
    \begin{tikzpicture} 
          \node[circle, draw, line width=0.6pt, inner sep=-0.4pt] (1) {\includegraphics[width=16pt]{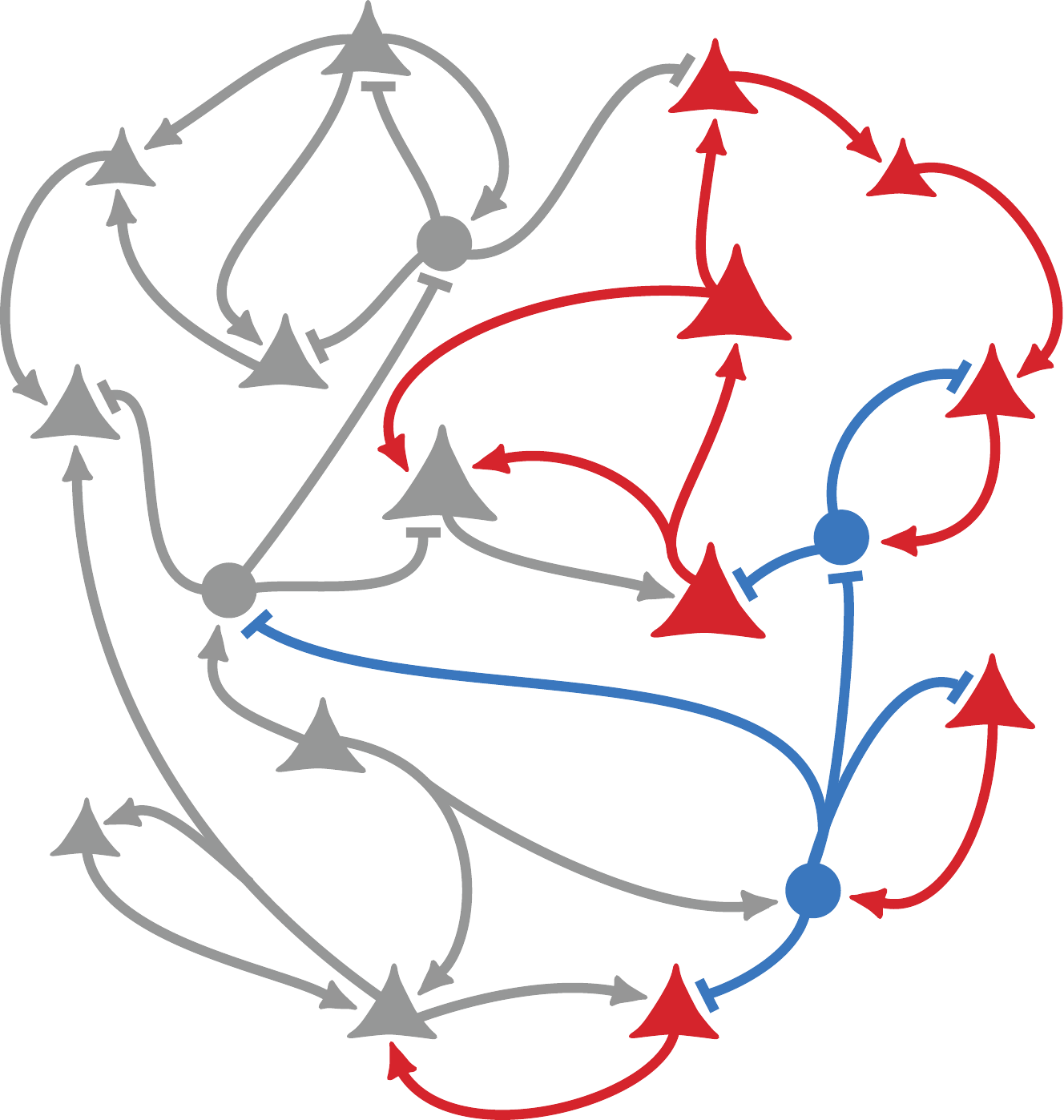}};
          \draw[-, line width=0.3pt] (1.north) to (1.south);
          \node[below of=1, yshift=-10pt, circle, draw, line width=0.6pt, inner sep=-0.4pt] (2) {\includegraphics[width=16pt]{EI_main_half_gray}};
          \draw[-, line width=0.3pt] (2.north) to (2.south);
          \node[below of=2, yshift=-10pt, circle, draw, line width=0.6pt, inner sep=-0.4pt] (3) {\includegraphics[width=16pt]{EI_main_half_gray}};
          \draw[-, line width=0.3pt] (3.north) to (3.south);
          \draw[latex-latex, line width=0.3pt, shorten <=1pt, shorten >=1pt] (1.270) to (2.90);
          \draw[latex-latex, line width=0.3pt, shorten <=1pt, shorten >=1pt] (2.270) to (3.90);
          \node[below of=3, yshift=10pt] {$\vdots$};
          \node[above of=1, yshift=-5pt] {$\vdots$};
      \node[left of=2, xshift=-20pt, yshift=40pt, scale=0.8] (i-1) {Subnetwork $i - 1$};
      \node[below of=i-1, yshift=-11pt, scale=0.8] (i) {Subnetwork $i$};
      \node[below of=i, yshift=-12pt, scale=0.8] (i+1) {Subnetwork $i + 1$};
      \node[right of=2, xshift=50pt, circle, draw, line width=1pt, inner sep=-2pt] (i-big) {\includegraphics[width=40pt]{EI_main_half_gray}};
      \draw[-] (i-big.north) to (i-big.south);
      \draw[shorten <=5pt, shorten >=25pt] (2.75) to (i-big.85);
      \draw[shorten <=5pt, shorten >=25pt] (2.285) to (i-big.275);
      \node[below of=i-big, xshift=-30pt, yshift=-25pt, scale=0.8] (i1) {\parbox{40pt}{\centering $\Nc_i^\ssm$ \\ (inhibited)}};
      \node[below of=i-big, xshift=30pt, yshift=-25pt, scale=0.8] (i2) {\parbox{40pt}{\centering $\Nc_i^\ssp$ \\ (recruited)}};
      \draw[-latex, bend right=0, shorten <=-8pt, shorten >=3pt] (i1.45) to (i-big.250);
      \draw[-latex, bend right=0, shorten <=-8pt, shorten >=3pt] (i2.135) to (i-big.290);
    \end{tikzpicture}
  \end{center} 
  \caption{The hierarchical network structure considered in this
    work. Each layer is only directly connected to the layers below
    and above it. Longer-range connections between non-successive
    layers do exist in thalamocortical hierarchies but are weaker than
    those between successive layers and are not considered in this
    work for simplicity.}
    \label{fig:multi}
\vspace*{-1ex}
\end{figure}

Even when considered in isolation, each layer of the network exhibits
rich dynamical behavior.  In fact, simulations of~\eqref{eq:dyn} with
random $\Wbf$ and $\dbf$ reveal that
\begin{itemize}
\item[--] locally, the dynamics may have zero, one, or many stable
  and/or unstable equilibrium points,
\item[--] globally, the dynamics can exhibit nonlinear phenomena such
  as limit cycles, multi-stability, and chaos,
\item[--] the state trajectories may grow unbounded (if $\mbf = \infty
  \ones_n$) if the excitatory subnetwork $[\Wbf]_0^\infty$ is
  sufficiently strong.
\end{itemize}
This richness of behavior can only increase if layers are subject to
time-varying inputs $\dbf(t)$ and, in particular, when interconnected
with other layers in the hierarchy.  Motivated by these observations,
our ultimate goal in this work is to tackle four problems:
\begin{enumerate}[wide]
\item the analysis of the relationship between structure ($\Wbf$) and
  dynamical behavior (basic properties such as existence and
  uniqueness of equilibria (EUE), asymptotic stability, and
  boundedness of trajectories) for each subnetwork when operating in
  isolation from the rest of the network ($\dbf(t) \equiv \dbf$);
\item the analysis of the conditions on the joint structure of each
  two successive layers $\Nc_i$ and $\Nc_{i+1}$ that allows for
  selective inhibition of $\Nc_{i+1}^\ssm$ by its input from $\Nc_i$,
  being equivalent to the stabilization of $\Nc_{i+1}^\ssm$ to $\zeros$
  (inactivity);
\item the analysis of the conditions on the joint structure of each
  two successive layers $\Nc_i$ and $\Nc_{i+1}$ that allows for
  top-down recruitment of $\Nc_{i+1}^\ssp$ by its input from $\Nc_i$,
  being equivalent to the stabilization of $\Nc_{i+1}^\ssp$ toward a
  desired trajectory set by $\Nc_i$ (activity);
\item the combination of (ii) and (iii) in a unified framework and its
  extension to the complete $N$-layer network of networks.
\end{enumerate} 
Problems (i) and (ii) are the focus of this paper, whereas we address
problems (iii) and (iv) in the accompanying
work~\cite{EN-JC:21-tacII}.

\section{Internal Dynamics of Single-Layer
  Networks}\label{sec:stability}

In this section, we provide an in-depth study of the basic dynamical
properties of the network dynamics~\eqref{eq:dyn} in isolation.  In
such case, the external input $\dbf(t)$ boils down to background
activity and possibly nonzero thresholds, which are constant relative
to the timescale~$\tau$. The dynamics~\eqref{eq:dyn} thus simplify to
\begin{align}\label{eq:dyn-d}
  \tau \dot \xbf(t) = -\xbf(t) + [\Wbf \xbf(t) + \dbf]_\zeros^\mbf,
  \quad &\zeros \le \xbf(0) \le \mbf,
  \\
  \notag &\mbf \in \realpos^n \cup \{\infty\}^n.
\end{align}
In the following, we derive conditions in terms of the network
structure for the existence and uniqueness of equilibria (EUE),
asymptotic stability, and boundedness of trajectories.

\subsection{Dynamics as Switched Affine System}

The nonlinear dynamics~\eqref{eq:dyn-d} is a switched affine system
with $2^n$ modes if $\mbf = \infty \ones_n$ or $3^n$ modes if $\mbf <
\infty \ones_n$.  Each mode of this system corresponds to a switching
index $\sigmab = \sigmab(\xbf) \in \zls^n$, where for each $i \in
\until{n}$, $\sigma_i = 0$ if the node is inactive (i.e., $(\Wbf \xbf
+ \dbf)_i \le 0$), $\sigma_i = \ell$ if the node is in linear regime
(i.e., $(\Wbf \xbf + \dbf)_i \in [0, m_i)$), and $\sigma_i = \s$ if
the node is saturated (i.e., $(\Wbf \xbf + \dbf)_i \ge m_i$). Clearly,
the mode of the system is state-dependent and each switching index
$\sigmab \in \zls^n$ corresponds to a switching region
\begin{align*}
  \Omega_\sigmab = \{\xbf \; | \; &(\Wbf \xbf + \dbf)_i \in (-\infty, 0], && \forall i \ \ {\rm s.t.}
  \ \ \sigma_i = 0, \text{ and}
  \\
  &(\Wbf \xbf + \dbf)_i \in [0, m_i], && \forall i \ \ {\rm s.t.} \ \
  \sigma_i = \ell, \text{ and}
  \\
  &(\Wbf \xbf + \dbf)_i \in [m_i, \infty), && \forall i \ \ {\rm s.t.}
  \ \ \sigma_i = \s\}.
\end{align*}
Within each $\Omega_\sigmab$, we have
\begin{align*}
  [\Wbf \xbf(t) + \dbf ]_\zeros^\mbf = \Sigmab^\ell (\Wbf \xbf(t) +
  \dbf) + \Sigmab^s \mbf,
\end{align*}
where $\Sigmab^\ell = \Sigmab^\ell(\sigmab)$ is a diagonal matrix with
$\Sigma^\ell_{ii} = 1$ if $\sigma_i = \ell$ and $\Sigma^\ell_{ii} = 0$
otherwise. $\Sigmab^\s$ is defined similarly, and we set the
convention that $\Sigmab^\s \mbf = \zeros$ if $\mbf = \infty
\ones_n$. Therefore, \eqref{eq:dyn-d} can be written in the equivalent
piecewise-affine form
\begin{align}\label{eq:dyn-switch}
  \tau \dot \xbf = (-\Ibf + \Sigmab^\ell \Wbf) \xbf + \Sigmab^\ell \dbf + \Sigmab^\s \mbf, \qquad
  \forall \xbf \in \Omega_\sigmab.
\end{align}
This switched representation of the dynamics motivates the following
assumptions on the weight matrix~$\Wbf$.

\begin{assumption}\label{as:1}
  Assume
  \begin{enumerate}
  \item $\det(\Wbf) \neq 0$;
  \item $\det(\Ibf - \Sigmab^\ell \Wbf) \neq 0$ for all the $2^n$
    matrices $\Sigmab^\ell(\sigmab), \sigmab \in \zls^n$.
  \end{enumerate}
\end{assumption}

Assumption~\ref{as:1} is not a restriction in practice since the set
of matrices for which it is not satisfied can be expressed as a finite
union of measure-zero sets, and hence has measure zero.  By
Assumption~\ref{as:1}(i), the system of equations $\Wbf \xbf + \dbf =
\zeros$ defines a non-degenerate set of $n$ hyperplanes partitioning
$\real^n$ into $2^n$ solid convex polytopic translated cones apexed at
$-\Wbf^{-1} \dbf$.%
\footnote{Recall that a set of $n$ hyperplanes is
  \emph{non-degenerate}~\cite{RTR-RJBW:98} if their intersection is a
  point or, equivalently, the matrix composed of their normal vectors
  is nonsingular.  A set $S \subseteq \real^n$ is called a
  \emph{polytope} if it has the form $S = \setdef{\xbf}{\Abf \xbf \le
    \bbf}$; a \emph{cone} if $c \xbf \in S$ for any $\xbf \in S, c \in
  \realnonneg$; a \emph{translated cone apexed at $\ybf$} if
  $\setdef{\xbf\!}{\!\xbf + \ybf \in S}$ is a cone; \emph{convex} if
  $(1 - \theta) \xbf + \theta \ybf \in S$ for any $\xbf, \ybf \in S,
  \theta \in [0, 1]$; and \emph{solid} if it has a non-empty interior.
}

Unlike linear systems, the existence of equilibria is not guaranteed
for~\eqref{eq:dyn-switch}.  Rather, each $\sigmab \in \zls^n$
corresponds to an \emph{equilibrium candidate}
\begin{align}\label{eq:h-uniq}
  \xbf^*_\sigmab = (\Ibf - \Sigmab^\ell \Wbf)^{-1} (\Sigmab^\ell \dbf + \Sigmab^\s \mbf),
\end{align}
which is an equilibrium if it belongs to $\Omega_\sigmab$.
We next identify conditions for this to be true.

\subsection{Existence and Uniqueness of Equilibria (EUE)}\label{subsec:eq}

The first step in analyzing any dynamical system is the
characterization of its equilibria. We being our analysis of the EUE
with the case of bounded activation functions ($\mbf \in \realpos^n$).

\begin{theorem}\longthmtitle{EUE}\label{thm:eue}
  The network dynamics~\eqref{eq:dyn-d} has a unique equilibrium for
  all $\dbf \in \real^n$ if and only if $\Ibf - \Wbf \in \Pc$.
\end{theorem}
\begin{proof}
  Despite their similarity, different equivalences need to be
  established and results need to be invoked for the bounded and
  unbounded cases. Therefore, we prove the result separately for each
  case.

  \emph{Case 1: $\mbf < \infty \ones_n$.} The existence of equilibria
  is guaranteed by the Brouwer fixed point theorem~\cite{LEJB:11} for
  all $\Wbf$ and all $\dbf$. We use results from~\cite{DK-RL:87} to
  characterize uniqueness. Following the terminology therein, the set
  $\Cc = \setdef{\Omega_\sigmab}{\sigmab \in \zls^n}$ is a chamber
  system and its branching number is $4$ by Assumption~\ref{as:1}. Let
  $f(\xbf; \dbf) = \xbf - [\Wbf \xbf + \dbf]_\zeros^\mbf$ which, for
  any $\dbf$, is piecewise-affine on the chamber system $\Cc$
  by~\eqref{eq:dyn-switch} and is proper since $\|f(\xbf; \dbf)\| \to
  \infty$ whenever $\|\xbf\| \to \infty$.

  First, assume that $\Ibf - \Wbf \in \Pc$. Then, $f$ is coherently
  oriented by definition and thus~\cite[Thm 5.3]{DK-RL:87} ensures
  that $f$ is bijective. In particular, there exists a unique $\xbf$
  such that $f(\xbf; \dbf) = \zeros$, giving uniqueness of equilibria
  for any $\dbf$.

  Now, assume that $\Ibf - \Wbf \notin \Pc$. Since the determinant of
  $\Ibf - \Sigmab^\ell \Wbf$ is always positive on the chamber
  $\Omega_\zeros$, $f$ cannot be coherently oriented, thus not
  bijective by~\cite[Thm 5.3]{DK-RL:87}, and thus not injective
  by~\cite[Cor 5.2]{DK-RL:87}. Therefore, there exists $\xbf_1,
  \xbf_2, \zbf \in \real^n$ such that $\xbf_1 \neq \xbf_2$ but
  \begin{align*}
    \xbf_1 - [\Wbf \xbf + \dbf]_\zeros^\mbf = \zbf = \xbf_2 - [\Wbf
    \xbf + \dbf]_\zeros^\mbf,
  \end{align*}
  where $\dbf \in \real^n$ is arbitrary. Therefore, $f(\, \cdot \,;
  \Wbf \zbf + \dbf)$ has two distinct roots, $\xbf_1 - \zbf$ and
  $\xbf_2 - \zbf$, proving the necessity of $\Ibf - \Wbf \in \Pc$ for
  uniqueness of equilibria.

  \emph{Case 2: $\mbf = \infty \ones_n$.} In this case, we simply show
  the equivalence between our equilibrium equation $\xbf = [\Wbf \xbf
  + \dbf]_\zeros^\infty$ and the well-studied linear complementarity
  problem (LCP). By Assumption~\ref{as:1},
  \begin{align}\label{eq:eq-mult}
    \xbf = [\Wbf \xbf + \dbf]_\zeros^\infty \Leftrightarrow \Wbf \xbf
    + \dbf = \Wbf [\Wbf \xbf + \dbf]_\zeros^\infty + \dbf.
  \end{align}
  We next perform a change of variables as follows. Let
  \begin{align*}
    \zbf = [\Wbf \xbf + \dbf]_\zeros^\infty \qquad \text{and} \qquad
    \wbf = [-\Wbf \xbf - \dbf]_\zeros^\infty.
  \end{align*}
  These vectors have the properties that
  \begin{align*}
    \zbf, \wbf \in \realnonneg^n, \qquad \zbf^T \wbf = 0, \quad
    \text{and} \quad \Wbf \xbf + \dbf = \zbf - \wbf.
  \end{align*}
  and thus provide a unique (invertible) characterization of
  $\xbf$. Therefore, \eqref{eq:eq-mult} is equivalent to
  \begin{align*}
    \wbf = (\Ibf - \Wbf) \zbf - \dbf, \qquad \zbf, \wbf \in
    \realnonneg^n, \qquad \zbf^T \wbf = 0,
  \end{align*}
  which is the standard LCP and has a unique solution $(\zbf, \wbf)$
  for all $\dbf \in \real^n$ if and only if $\Ibf - \Wbf \in
  \Pc$~\cite{KGM:72}. Similar to Case 1, it can also be shown that if
  $\Ibf - \Wbf \notin \Pc$, there exists at least one $\dbf$ for which
  two equilibrium points exists (see, e.g., the proof of~\cite[Thm
  4.2]{KGM:72}). This completes the proof.
\end{proof} 

Even though the condition $\Ibf - \Wbf \in \Pc$ may seem abstract, it
has a nice geometric interpretation. From~\cite[Lem 2.2]{DK-RL:87},
$\Ibf - \Wbf \in \Pc$ if and only if the (negative) vector field $\xbf
\mapsto \xbf - [\Wbf \xbf + \dbf]_\zeros^\mbf$ maps each switching
region $\Omega_\sigma$ to another polytopic region and the images of
adjacent switching regions remains adjacent. In other words, the
vector field has a \emph{coherent orientation} when mapping the state
space.%
\footnote{A closely-related class of matrices is that of
  M-matrices~\cite{DSB:09} with established relationships with the
  stability of nonlinear systems, see, e.g.,~\cite{DMS-DDS:00}.}

\begin{remark}\longthmtitle{Computational complexity of verifying
    $\Ibf - \Wbf \in \Pc$}\label{rem:p}
  {\rm Although the problem of determining whether a matrix is in
    $\Pc$ is straightforward for small $n$, it is known to be
    co-NP-complete~\cite{GEC:94}, and thus expensive for large
    networks.  Indeed,~\cite{SMR:03} shows that all the $2^n$
    principal minors of $\Abf$ have to be checked to prove $\Abf \in
    \Pc$ (though disproving $\Abf \in \Pc$ is usually much easier).
    In these cases, one may need to rely on more conservative
    sufficient conditions such as $\rho(|\Wbf|) < 1$ or $\|\Wbf\| < 1$
    (cf. Lemma~\ref{lem:inc-mat-class}) to establish $\Ibf - \Wbf \in
    \Pc$.  These conditions, moreover, have the added benefit of
    providing intuitive connections between the distribution of
    synaptic weights, network size, and stability. We elaborate more
    on this point in Section~\ref{subsec:network-size}.  \oprocend }
\end{remark}

\begin{example}\longthmtitle{Wilson-Cowan model}\label{ex:ei}
  {\rm Consider a network of $n$ nodes in which $\alpha n, \alpha \in
    (0, 1)$ are excitatory, $(1 - \alpha)n$ are inhibitory. Under some
    regularity assumptions, given next, this network can be (further)
    reduced to a simple, two-dimensional network commonly known as the
    Wilson-Cowan model~\cite{HRW-JDC:72} and widely used in
    computational
    neuroscience~\cite{ACEO-MWJ-RB:14,MPJ-TJS:14}. Assume that the
    synaptic weight between any pair of nodes, the external input to
    them, and their maximal firing rate (if finite) only depends on
    their type: the synaptic weight of any inhibitory-to-excitatory
    connection is $w_{ei} < 0$, similarly for $w_{ee} > 0, w_{ie} > 0,
    w_{ii} < 0$, and all excitatory (inhibitory) nodes receive $d_e
    \in \real$ ($d_i \in \real$) and have maximal rate $m_e \in
    \realpos \cup \{\infty\}$ ($m_i \in \realpos \cup \{\infty\}$).
    Let $x_e(t)$ and $x_i(t)$ be the average firing rates of
    excitatory and inhibitory nodes, respectively. Then,
    \eqref{eq:dyn-d} simplifies to
    \begin{align*}
      \!\tau \begin{bmatrix} \dot x_e \\ \dot x_i \end{bmatrix} =
      -\!\begin{bmatrix} x_e \\ x_i \end{bmatrix} +
      \left[\begin{bmatrix} \alpha n w_{ee} & \!\!(1 \!-\! \alpha) n
          w_{ei} \\ \alpha n w_{ie} &
          \!\!(1 \!-\! \alpha) n w_{ii} \end{bmatrix} \! \begin{bmatrix} x_e \\
          x_i \end{bmatrix} + \begin{bmatrix} d_e \\
          d_i \end{bmatrix} \right]_\zeros^\mbf.
    \end{align*}
    Let $\Wbf_{EI} \in \real^{2
    \times 2}$ be the corresponding weight matrix above.  One can
  check that
  \begin{align}\label{eq:wceue}
    &\Ibf - \Wbf_{EI} \in \Pc \Leftrightarrow \alpha n w_{ee} < 1,
    \end{align}
    and
    \begin{align*}
    &\rho(|\Wbf_{EI}|) < 1 \Leftrightarrow \alpha n w_{ee} < 1 
    , \ (1 - \alpha) n |w_{ii}| < 1, \ \text{and}
    \\
    &\alpha(1 - \alpha) n^2 w_{ie} |w_{ei}| < (1 - \alpha
    n w_{ee})(1 - (1 - \alpha) n |w_{ii}|).
  \end{align*}
  Thus, according to Theorem~\ref{thm:eue}, EUE only requires the
  excitatory dynamics to be stable (note that $w_{ee}$ has to become
  smaller as $n$ grows), while the more conservative condition
  $\rho(|\Wbf_{EI}|) < 1$ also requires (relatively) weak
  inhibitory-inhibitory synapses and a weak interconnection between
  excitatory and inhibitory subnetworks. \oprocend }
\end{example}

When $\Ibf - \Wbf \in \Pc$, Theorem~\ref{thm:eue} ensures EUE for all
$\dbf \in \real^n$. When $\Ibf - \Wbf \notin \Pc$, however, a more
involved question is to find the values of $\dbf$ that give rise to
non-unique equilibrium points. To answer this, we next perform  a more direct
analysis of the equilibria. For simplicity, we focus in the
remainder of this section on unbounded dynamics~($\mbf = \infty
\ones_n$).

Recall the definition of an equilibrium candidate
in~\eqref{eq:h-uniq}.  Using Assumption~\ref{as:1}, and after some
manipulations, we have
\begin{align}\label{eq:auxx}
  \Wbf \xbf^*_\sigmab + \dbf &= \Wbf (\Ibf - \Sigmab^\ell \Wbf)^{-1} \Sigmab^\ell \dbf + \dbf 
  \\
  &= (\Wbf^{-1} -
  \Sigmab^\ell)^{-1} \Sigmab^\ell \dbf + \dbf \nonumber
  \\
  &= (\Ibf - \Wbf \Sigmab^\ell)^{-1} \Wbf \Sigmab^\ell \dbf + \dbf \nonumber
  \\
  &= \big[(\Ibf - \Wbf \Sigmab^\ell)^{-1} \Wbf \Sigmab^\ell + \Ibf\big] \dbf = (\Ibf - \Wbf
  \Sigmab^\ell)^{-1} \dbf,  \nonumber
\end{align}
thus,
\begin{align}\label{eq:candid-cond}
  \xbf^*_\sigmab \in \Omega_\sigmab \Leftrightarrow \underbrace{(2
    \Sigmab^\ell - \Ibf) (\Ibf - \Wbf \Sigmab^\ell)^{-1}}_{\triangleq
    \Mbf_\sigmab} \dbf \ge \zeros.
\end{align}
Therefore, if $\Mbf_\sigmab \dbf \ge \zeros$ for exactly one $\sigmab
\in \{0, \ell\}^n$, then a unique equilibrium exists. 
However, when $\Mbf_{\sigmab_\ell} \dbf \ge \zeros$ for multiple
$\sigmab_\ell \in \{0, \ell\}^n, \ell \in \until{\bar \ell}$, the
network may have either multiple equilibria or a unique one
$\xbf^*_{\sigmab_1} = \cdots = \xbf^*_{\sigmab_{\bar \ell}}$ lying on
the boundary between $\{\Omega_{\sigma_\ell}\}_{\ell = 1}^{\bar
  \ell}$. The next result shows that the quantities $\Mbf_\sigmab
\dbf$ can be used to distinguish between these two latter cases.

\begin{lemma}\longthmtitle{Existence of multiple
    equilibria}\label{lem:mul-eq} 
  Assume $\Wbf$ satisfies Assumption~\ref{as:1}, $\dbf \in \real^n$ is
  arbitrary, and $\Mbf_\sigmab$ is defined as
  in~\eqref{eq:candid-cond} for $\sigmab \in \{0, \ell\}^n$.  If there
  exist $\sigmab_1 \neq \sigmab_2$ such that $\Mbf_{\sigmab_1} \dbf
  \ge \zeros$ and $\Mbf_{\sigmab_2} \dbf \ge \zeros$,
  then $\xbf^*_{\sigmab_1} = \xbf^*_{\sigmab_2}$ if and only if
  $\Mbf_{\sigmab_1} \dbf = \Mbf_{\sigmab_2} \dbf$.
\end{lemma}
\begin{proof}
  Clearly,
  \begin{align}\label{eq:mul-eq}
    \notag \xbf^*_{\sigmab_1} = \xbf^*_{\sigmab_2} &\Leftrightarrow \Wbf
    \xbf^*_{\sigmab_1} + \dbf = \Wbf \xbf^*_{\sigmab_2} + \dbf
    \\
    &\Leftrightarrow (\Ibf - \Wbf \Sigmab_1)^{-1} \dbf = (\Ibf - \Wbf \Sigmab_2)^{-1}
    \dbf,
  \end{align}
  where we have used~\eqref{eq:auxx}.  Since both $\Mbf_{\sigmab_1} \dbf$ and
  $\Mbf_{\sigmab_2} \dbf$ are nonnegative, \eqref{eq:mul-eq} holds if and
  only if $\big((\Ibf - \Wbf \Sigmab_1)^{-1} \dbf\big)_i = \big((\Ibf - \Wbf
  \Sigmab_2)^{-1} \dbf\big)_i = 0$ for any $i$ such that $ \sigma_{1, i}
  \neq \sigma_{2, i}$, which is equivalent to $\Mbf_{\sigmab_1} \dbf =
  \Mbf_{\sigmab_2} \dbf$.
\end{proof}

This property of $\Mbf_\sigmab$ can be used to derive a
computationally more involved but input-dependent characterization of
EUE, as follows.

\begin{proposition}\longthmtitle{Optimization-based condition for
    EUE}\label{prop:opt-cond} 
  Let $\Wbf$ satisfy Assumption~\ref{as:1} and $\Mbf_\sigmab$ be as
  defined in~\eqref{eq:candid-cond} for $\sigmab \in \{0, \ell\}^n$. 
   For $\dbf \in
  \real^n$, define $\mu_1(\dbf)$ and $\mu_2(\dbf)$ to be the largest
  and second largest elements of the set
  \begin{align*}
    \setdefb{\min_{i = 1, \dots, n} (\Mbf_\sigmab \dbf)_i}{\sigmab
      \in \{0, \ell\}^n},
  \end{align*}
  respectively. Then, \eqref{eq:dyn-d} has a unique equilibrium for each $\dbf
  \in \real^n$ if and only if
  \begin{align}\label{eq:opt-cond}
    \max_{\|\dbf\| = 1} \mu_1(\dbf) \mu_2(\dbf) < 0.
  \end{align}
\end{proposition}
\begin{proof}
  First, note that $\dbf = \zeros$ is a degenerate case where the
  origin is the unique equilibrium belonging to all
  $\Omega_\sigmab$. For any $\dbf \neq \zeros$ and $\sigmab \in \{0,
  \ell\}^n$, $\Mbf_\sigmab \dbf \ge \zeros$ if and only if
  $\Mbf_\sigmab \dbf / \|\dbf\| \ge \zeros$. Thus, EUE for all $\dbf
  \in \real^n$ and all $\|\dbf\| = 1$ are equivalent.  Then, for any
  $\dbf$,
  \begin{align}\label{eq:mus}
    \notag \!\!\!\!\!\!\mu_1(\dbf) \mu_2(\dbf) < 0 &\Leftrightarrow
    \mu_1(\dbf) > 0 \ \text{and} \ \mu_2(\dbf) < 0
    \\
    &\Leftrightarrow \exists \ \text{unique} \ \Mbf_\sigmab \dbf \ge
    \zeros, \quad \sigmab \in \{0, \ell\}^n. \!\!\!\!
  \end{align}
  Note that the latter allows for the possibility of the existence of
  multiple $\sigmab$ with $\Mbf_\sigmab \dbf \ge \zeros$ provided that
  they have the same value of $\Mbf_\sigmab \dbf$. By
  Lemma~\ref{lem:mul-eq}, \eqref{eq:mus} is then equivalent to EUE,
  completing the proof.
\end{proof} 

In our experience, the optimization involved in~\eqref{eq:opt-cond} is
usually highly non-convex but since the search space $\{\|\dbf\| =
1\}$ is compact, global search methods can be used to
verify~\eqref{eq:opt-cond} numerically if $n$ is not too large.
However, note that our main interest in~\eqref{eq:opt-cond} (being
equivalent to $\Ibf - \Wbf \in \Pc$) is when it does \emph{not}
hold. If so, then any $\dbf$ for which~\eqref{eq:opt-cond} fails gives
a ray $\setdef{\alpha \dbf}{\alpha > 0}$ of input values that give
rise to non-unique equilibria. Combined with stability analysis of
Section~\ref{subsec:as}, e.g., this can be a basis for the
characterization of multistability in linear-threshold dynamics which
is itself beyond the scope of this work.

The proof of Theorem~\ref{thm:eue} (for the unbounded case) is based
on the LCP, which makes the relationship between $\Ibf - \Wbf \in \Pc$
and EUE opaque, even when taking into account the proof of the
LCP. The equilibrium characterization in~\eqref{eq:candid-cond},
however, can be used to explain this relationship more
transparently. For any given $\dbf$, the non-uniqueness of equilibria
is equivalent to asking whether
\begin{align*}
  \exists \sigmab_1, \sigmab_2 \in \{0, \ell\}^n \quad \text{s.t.}
  \quad &\Mbf_{\sigmab_1} \dbf \ge \zeros \ \ \text{and} \ \
  \Mbf_{\sigmab_2} \dbf \ge \zeros
  \\
  &\Mbf_{\sigmab_1} \dbf \neq \Mbf_{\sigmab_2} \dbf,
\end{align*}
or, whether there exist $\qbf_{\sigmab_1} \neq \qbf_{\sigmab_2} \in
\realnonneg^n $ such that $\Mbf_{\sigmab_1}^{-1} \qbf_{\sigmab_1} =
\Mbf_{\sigmab_2}^{-1} \qbf_{\sigmab_2} = \dbf$.  A more general
question, which turns out to be relevant to EUE, is whether
\begin{align}\label{eq:mxmy}
  \exists \qbf_{\sigmab_1} \neq \qbf_{\sigmab_2} \in \Oc_n \quad
  \text{s.t.} \quad \qbf_{\sigmab_1} = \Mbf_{\sigmab_1}
  \Mbf_{\sigmab_2}^{-1} \qbf_{\sigmab_2},
\end{align}
for \emph{any} orthant $\Oc_n$ of $\real^n$ (including $\Oc_n =
\realnonneg^n$ as a special case).  This depends on whether the matrix
$\Mbf_{\sigmab_1} \Mbf_{\sigmab_2}^{-1}$ can map any nonzero vector to
the same orthant which, by Lemma~\ref{lem:p-mat}(iii), happens if and
only if $-\Mbf_{\sigmab_1} \Mbf_{\sigmab_2}^{-1} \notin \Pc$. The
following result, whose proof is in ~\ref{app:proofs}, gives a
necessary and sufficient condition for this to not happen.
  
\begin{theorem}\longthmtitle{Coherently oriented vector fields and the
    validity of equilibrium candidates}\label{thm:eq-candid-P}
  Let $\Wbf$ satisfy Assumption~\ref{as:1} and $\Mbf_\sigma$ be
  defined as in~\eqref{eq:candid-cond}. Then, $-\Mbf_{\sigmab_1}
  \Mbf_{\sigmab_2}^{-1} \in \Pc$ for all (distinct) $\sigmab_1,
  \sigmab_2 \in \{0, \ell\}^n$ if and only if $\Ibf - \Wbf \in \Pc$.
\end{theorem}

Theorem~\ref{thm:eq-candid-P} provides a more transparent account of
the relationship between $\Ibf - \Wbf \in \Pc$ and EUE. If $\Ibf -
\Wbf \in \Pc$, then Theorem~\ref{thm:eq-candid-P} and
Lemma~\ref{lem:p-mat}(iii) ensure that none of $\Mbf_{\sigmab_1}
\Mbf_{\sigmab_2}^{-1}$ can map a vector to the same orthant. Thus, no
two $\qbf_\sigmab = \Mbf_\sigmab \dbf$ belong to the same
orthant. Therefore, there exists a one-to-one correspondence between
$\{\qbf_\sigmab\}$ and orthants in $\real^n$, ensuring that exactly
one $\qbf_\sigmab$ belongs to $\realnonneg^n$, i.e., EUE.%
\footnote{With a careful resolution of the ties, this still holds in
  the measure-zero event that multiple $\qbf_\sigmab$ are equal and
  belong to the boundary between orthants.}

We end this subsection with a result that bounds the number and
location if equilibria for the case when $\Ibf - \Wbf \notin \Pc$. For
$\Abf \in \real^{n \times n}$ and $\sigmab \in \{0, \ell\}^n$, let
$\Abf_{(\sigmab)}$ be the principal submatrix of $\Abf$ containing the
rows and columns for which $\sigma_i = \ell$. Further, for $\sigmab_1,
\sigmab_2 \in \{0, \ell\}^n$, we say $\sigmab_1 \le \sigmab_2$ if
$\sigma_{1, i} = \ell \Rightarrow \sigma_{2, i} = \ell$ for all $i \in
\until{n}$.

\begin{corollary}\longthmtitle{Partial EUE}\label{cor:peue}
  Consider the dynamics~\eqref{eq:dyn-d} and assume that
  Assumption~\ref{as:1} holds. If $\Ibf - \Wbf_{(\bar \sigmab)} \in
  \Pc$ for any $\bar \sigmab \in \{0, \ell\}^n$, then
  $\bigcup_{\sigmab \le \bar \sigmab} \Omega_\sigmab$ contains at most
  one equilibrium point.
\end{corollary}
\begin{proof}
  The proof follows directly from the proof of
  Theorem~\ref{thm:eq-candid-P} and the fact that, using the
  definitions therein, $-\Gammab \in \Pc$ only requires $\Ibf -
  \Wbf_{([\ellb_{n_1+n_2+n_3} \ \zeros_{n_4}]^T)} \in \Pc$.
\end{proof} 

Even in the simplest case when $\Ibf - \Wbf \in \Pc$, the resulting
unique equilibrium may or may not be stable, as studied next.

\subsection{Asymptotic Stability}\label{subsec:as}

The EUE is an \emph{opportunity} to shape the network activity at
steady state, provided that the equilibrium corresponds to a desired
state (a memory, the encoding of a spatial location, eye position,
etc.~\cite{JJH:82,HSS-DDL-BYR-DWT:00, DD-JKS-TJS:00,RC-DA-RY:03,
  JJK-KZ:12}) \emph{and} it attracts network trajectories.  Here we
investigate the latter.

\begin{theorem}\longthmtitle{Asymptotic Stability}\label{thm:ges}
  Consider the network dynamics~\eqref{eq:dyn-d} and assume $\Wbf$
  satisfies Assumption~\ref{as:1}.
  \begin{enumerate}
  \item {[Sufficient condition]} If $\Wbf \in \Lc$, then for all $\dbf
    \in \real^n$, the network is globally exponentially stable (GES)
    relative to a unique equilibrium $\xbf^*$;
  \item {[Necessary condition]} If for all $\dbf \in \real^n$ the
    network is locally asymptotically stable relative to a unique
    equilibrium $\xbf^*$, then $-\Ibf + \Wbf \in \Hc$.
  \end{enumerate}
\end{theorem}
\begin{proof}
  \emph{(i)} The EUE follows from
  Lemma~\ref{lem:inc-mat-class}(iii)\&(iv) and
  Theorem~\ref{thm:eue}. GES can be deduced from~\cite[Thm
  1]{AP-NvdW-HN:05}, but a simpler and direct proof can also be found
  in a preliminary version of this work~\cite{EN-JC:19-tacI-v2} which
  is omitted here space reasons.

  \emph{(ii)} Assume, by contradiction, that $-\Ibf + \Wbf \notin
  \Hc$, which means that there exists $\sigmab \in \{0, \ell\}^n$ such
  that $-\Ibf + \Sigmab^\ell \Wbf$ is not Hurwitz. Let $\sigmab_{01}
  \izon$ have the same zeros as $\sigmab$, and consider the choice
    \begin{align*}
      \dbf = (2 \Ibf - \Wbf) \sigmab_{01} - \ones_n.
    \end{align*}
    It is straightforward to show that $\xbf^* = \sigmab_{01}$ is an
    equilibrium point for~\eqref{eq:dyn-d} lying in the interior of
    $\Omega_\sigmab$. By assumption, $\xbf^*$ is (unique and) locally
    asymptotically stable, which contradicts $-\Ibf + \Sigmab^\ell \Wbf$
    not being Hurwitz. This completes the proof.
\end{proof}

Similar to the problem of verifying whether a matrix is a P-matrix,
cf. Remark~\ref{rem:p}, the computational complexity of verifying
total-Hurwitzness grows exponentially with $n$. The same applies to
the verification of total $\Lc$-stability, see, e.g., \cite{DL-RT:04}
and the references therein. The next result gives a usually more
conservative but computationally inexpensive alternative.

\begin{proposition}\longthmtitle{Computationally feasible sufficient
    conditions for GES}\label{prop:comp-ges}
  Consider the network dynamics~\eqref{eq:dyn-d} and assume the weight
  matrix $\Wbf$ satisfies Assumption~\ref{as:1}.  If $\rho(|\Wbf|) <
  1$ or $\|\Wbf\| < 1$, then for all $\dbf \in \real^n$, the network
  has a unique equilibrium $\xbf^*$ and it is GES relative to
  $\xbf^*$.
\end{proposition}
\begin{proof}
  If $\|\Wbf\| < 1$, the result follows from
  Lemma~\ref{lem:inc-mat-class}(ii) and Theorem~\ref{thm:ges}. For the
  case $\rho(|\Wbf|) < 1$, the same proof technique as
  in~\cite[Prop. 3]{JF-KPH:96} can be used to prove GES, as shown in a
  preliminary version of this work~\cite{EN-JC:19-tacI-v2}, but is
  omitted here for space reasons.
\end{proof}

From Lemma~\ref{lem:inc-mat-class}(iii), the conditions of
Theorem~\ref{thm:ges} and Proposition~\ref{prop:comp-ges} are not
conclusive when $\Wbf$ satisfies $-\Ibf + \Wbf \in \Hc$ but neither
$\Wbf \in \Lc$ nor $\rho(|\Wbf|) < 1$. However,
\begin{enumerate}[wide]
\item If a unique equilibrium $\xbf^*$ lies in the interior of a
  switching region (a condition that can be shown to hold for almost
  all $\dbf$), then $\xbf^*$ is at least locally exponentially stable
  (since a sufficiently small region of attraction of $\xbf^*$ is
  contained in that switching region).
\item In our extensive simulations with random $(\Wbf, \dbf)$, any
  system satisfying $-\Ibf + \Wbf \in \Hc$ was GES for all $\dbf$.
\end{enumerate}
These observations lead us to the following conjecture, whose analytic
characterization remains an open problem.

\begin{conjecture}\longthmtitle{Sufficiency of total-Hurwitzness for
    GES}\label{conj:ges} 
  Consider the dynamics~\eqref{eq:dyn-d} and assume $\Wbf$ satisfies
  Assumption~\ref{as:1}. The network has a unique GES equilibrium for
  all $\dbf \in \real^n$ if and only if $-\Ibf + \Wbf \in \Hc$.
\end{conjecture}

We next study the GES of the uniform excitatory-inhibitory networks of
Example~\ref{ex:ei}.

\begin{example}\longthmtitle{Wilson-Cowan model,
    cont'd}\label{ex:ei2}
  {\rm Consider the Wilson-Cowan model 
    of Example~\ref{ex:ei}. One can verify
    \begin{align}\label{eq:ei-h}
      -\Ibf + \Wbf_{EI} \in \Hc \Leftrightarrow \alpha n w_{ee} < 1,
    \end{align}
    thus being equivalent (in this two-dimensional case) to $\Ibf -
    \Wbf_{EI} \in \Pc$ and, interestingly, only restricting $w_{ee}$
    while $w_{ei}$, $w_{ie}$, and $w_{ii}$ are completely free.
    Figure~\ref{fig:ei} shows sample phase portraits for the cases
    $\alpha n w_{ee} < 1$ and $\alpha n w_{ee} > 1$, matching our
    expectations from Theorems~\ref{thm:eue} and~\ref{thm:ges}.  While
    our focus here is on the existence, uniqueness, and stability of
    equilibria, it is instructive to highlight the role of equilibrium
    analysis and, in particular, \emph{lack of stable equilibria} in
    the generation of oscillations in the same linear-threshold
    Wilson-Cowan model~\cite{EN-JC:19-acc}. In this case, both the
    linear-threshold Wilson-Cowan model and the popular Kuramoto
    model~\cite{YQ-YK-MC:18,MJ-XY-MP-HS-KH:18,TM-GB-DSB-FP:18} of
    neural oscillations provide parallel simplifications to the (more
    biologically accurate) Wilson-Cowan model with smooth sigmoidal
    nonlinearities, cf.~\cite{HGS-PW:90}.} \oprocend
\end{example}

\begin{figure}
  \centering \subfloat[]{\small
    \hspace*{-18pt} \includegraphics[width=.9\linewidth]{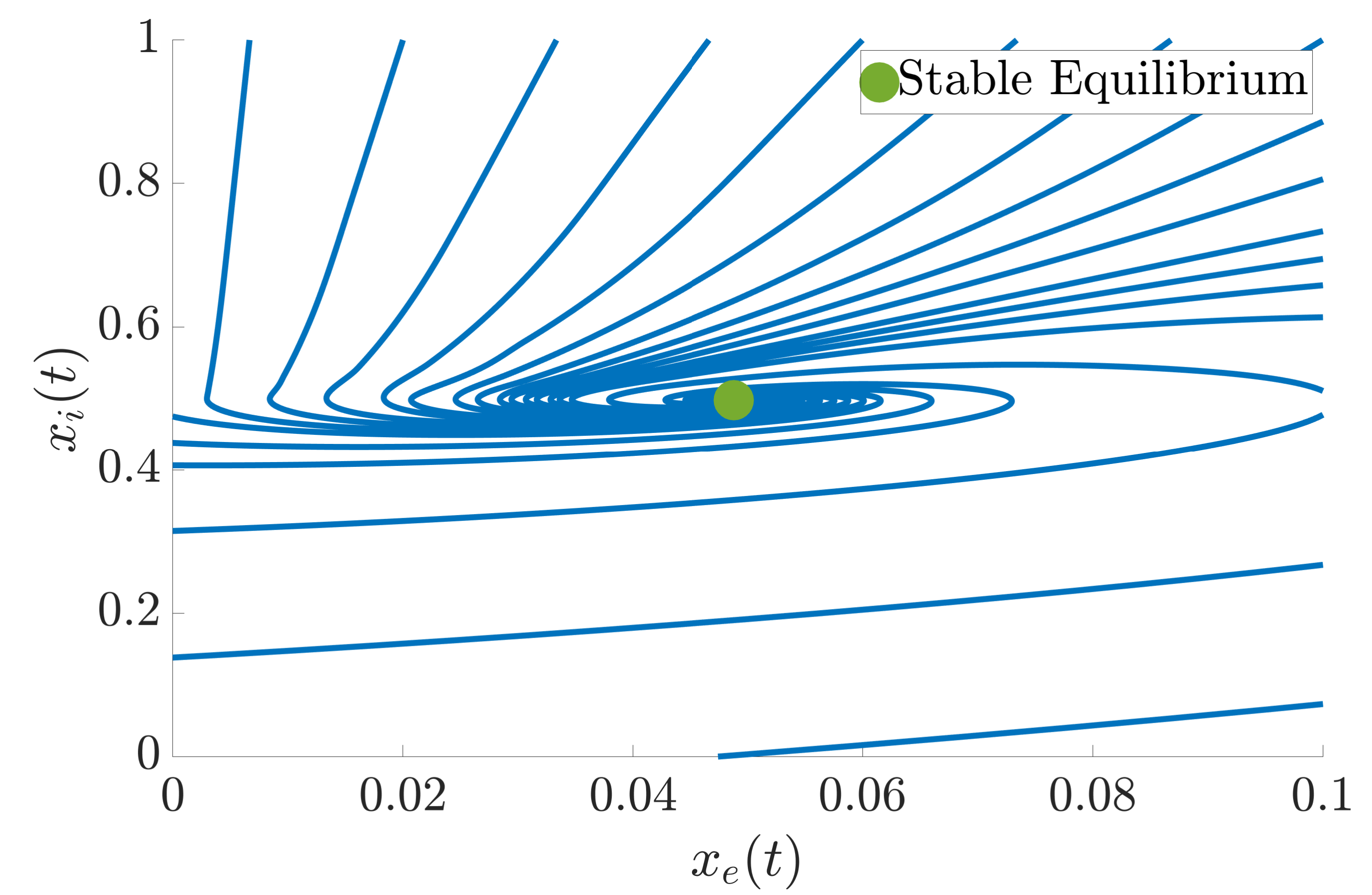}
  }
  \\
  \subfloat[]{\small
    \hspace*{-22pt} \includegraphics[width=0.9\linewidth]{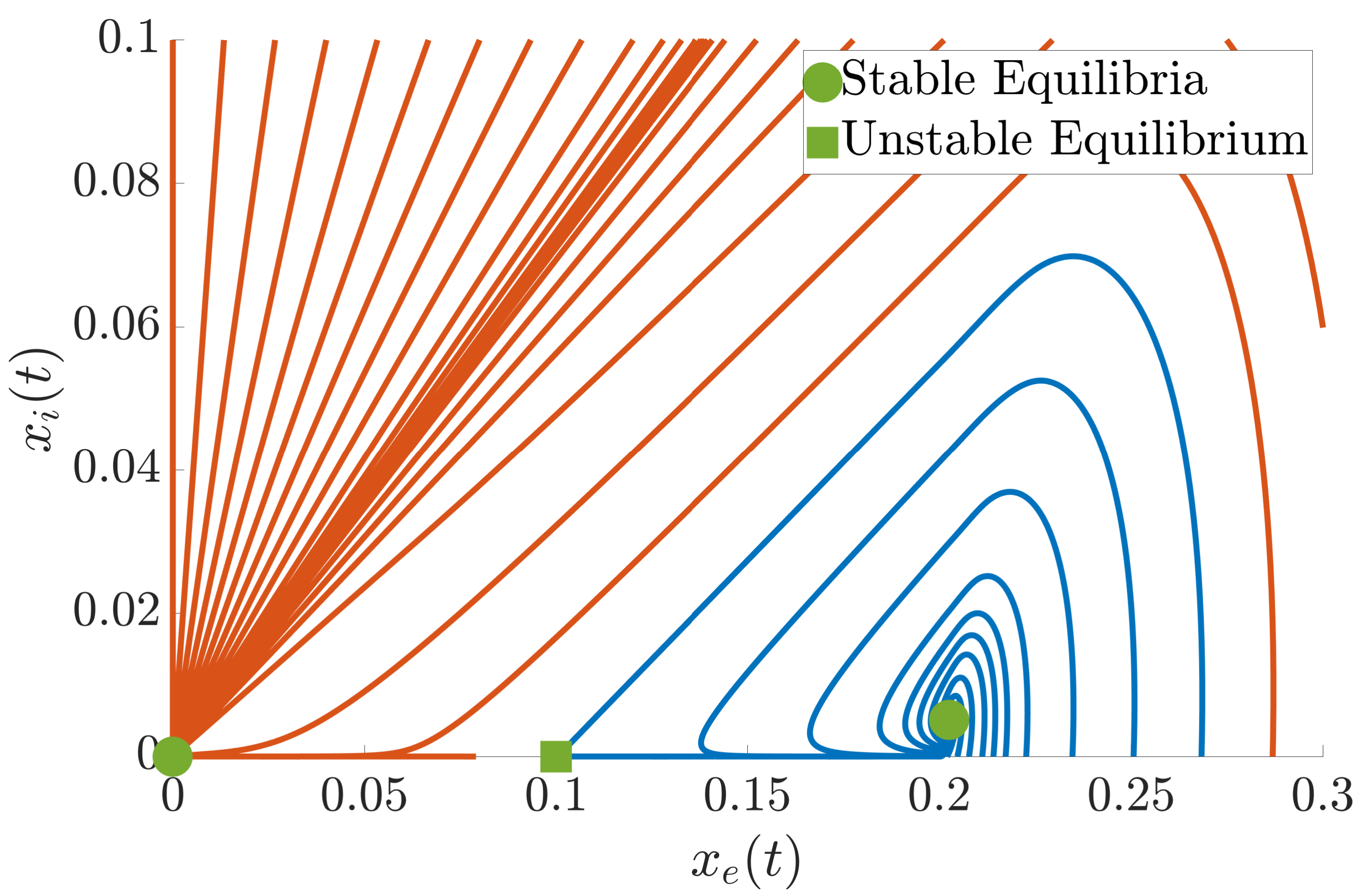}
  }
  \caption{\small Network trajectories for the excitatory-inhibitory
    network of Example~\ref{ex:ei2}.  \textbf{(a)} When $\Wbf_{EI} =
    [0.9, -2; 5, -1.5]$, $\dbf_{EI} = [1; 1]$, the network has a
    unique GES equilibrium.  \textbf{(b)} However, for $\Wbf_{EI} =
    [1.1, -2; 5, -1.5]$, $\dbf_{EI} = [-0.01; -1]$, the network
    exhibits bi-stable behavior. $\mbf = \infty \ones_2$ in both
    cases. The trajectory colors correspond to the equilibria to which
    they converge.  Although $\alpha n w_{ee} > 1$, the network is GES
    for most values of $\dbf_{EI}$, so we used
    Proposition~\ref{prop:opt-cond} for finding a $\dbf_{EI}$ that
    leads to multi-stability.}
  \label{fig:ei}
\vspace*{-1.5ex}
\end{figure}

\subsection{Boundedness of Solutions}\label{subsec:b}

Here we study the boundedness of solutions under the
dynamics~\eqref{eq:dyn}.  While our discussion so far has been
about~\eqref{eq:dyn-d} (with constant $\dbf$), we switch for the
remainder of this section to~\eqref{eq:dyn} for the sake of
generality, as the same results are applicable without major
modifications. Since network trajectories are trivially bounded if
$\mbf < \infty \ones_n$, we limit our discussion here to the unbounded
case. Note that in reality, the firing rate of any neuron is bounded
by a maximum rate dictated by its refractory period (the minimum
inter-spike duration).  Unboundedness of solutions in this model
corresponds in practice to the so-called ``run-away'' excitations
where the firing of neurons grow beyond sustainable rates for
prolonged periods of time, which is neither desirable nor
safe~\cite{PJ-JC-ADP-JEF-WCMV-XL-MP-AFB-RWD-JGRJ:10}.

Since GES implies boundedness of solutions, any condition that is
sufficient for GES is also sufficient for boundedness. However,
boundedness of solutions can be guaranteed under less restrictive
conditions. The next result shows that inhibition, overall, preserves
boundedness.

\begin{lemma}\longthmtitle{Inhibition preserves
    boundedness}\label{lem:inh-bdd}
  Let $t \mapsto \xbf(t)$ be the solution of~\eqref{eq:dyn} starting
  from initial state $\xbf(0) = \xbf_0$. Consider the system
  \begin{align}\label{eq:ib-xp}
    \tau \dot{\bar \xbf}(t) = -\bar \xbf(t) + [[\Wbf]_\zeros^\infty \bar \xbf(t) +
    \dbf(t)]_\zeros^\infty, \quad \bar \xbf(0) = \xbf_0.
  \end{align}
  Then, $\xbf(t) \le \bar \xbf(t)$ for all $t \ge 0$.
\end{lemma}
\begin{proof}
  Since $\xbf(t) \ge \zeros$ for all $t$, we can write~\eqref{eq:dyn} as
  \begin{align}\label{eq:ib-x2}
    \tau \dot \xbf(t) = -\xbf(t) + [[\Wbf]_\zeros^\infty \xbf(t) + \dbf(t) + \deltab(t)]_\zeros^\infty,
  \end{align}
  where $ \deltab(t) \triangleq (\Wbf - [\Wbf]_\zeros^\infty) \xbf(t)
  \le \zeros$.  Since the vector field $(\xbf, t) \mapsto -\xbf +
  [[\Wbf]_\zeros^\infty \xbf + \dbf(t)]_\zeros^\infty$ is
  quasi-monotone nondecreasing\footnote{A vector field $f: \real^n
    \times \real \to \real^n$ is \emph{quasi-monotone
      nondecreasing}~\cite[Def 2.3]{SKYN:13} if for any $\xbf, \ybf
    \in \real^n$ and any $i \in \until{n}$,
    \begin{align*}
      \big(x_i = y_i \ \ \text{and} \ \ x_j \le y_j \ \text{for all} \ j
      \neq i \big) \Rightarrow f(\xbf, t) \le f(\ybf, t).
    \end{align*}
  }, the result follows by using the monotonicity of $[\, \cdot\,
  ]_0^\infty$ and applying the vector-valued extension of the
  Comparison Principle given in~\cite[Lemma~3.4]{SKYN:13}
  to~\eqref{eq:ib-xp} and~\eqref{eq:ib-x2}.
\end{proof}

While the result about preservation of boundedness under inhibition in
Lemma~\ref{lem:inh-bdd} is intuitive, one must interpret it carefully:
it is \emph{not} in general true that adding inhibition to any
dynamics~\eqref{eq:dyn} can only decrease $\xbf(t)$. This is only true
if the network vector field is quasi-monotone nondecreasing, as is the
case with the excitatory-only dynamics~\eqref{eq:ib-xp}. Intuitively,
this is because, if the network has inhibitory nodes, adding
inhibition to their input can in turn ``disinhibit'' and increase the
activity of the rest of the network.  The next result identifies a
condition on the excitatory part of the dynamics to determine if
trajectories are bounded.

\begin{theorem}\longthmtitle{Boundedness}\label{thm:bdd}
  Consider the network dynamics~\eqref{eq:dyn}. If the
  corresponding excitatory-only dynamics~\eqref{eq:ib-xp} has bounded
  trajectories, the trajectories of~\eqref{eq:dyn} are also bounded by
  the same bound as those of~\eqref{eq:ib-xp}.
\end{theorem}

The proof of this result follows from Lemma~\ref{lem:inh-bdd} and is
therefore omitted. The following result, similar to
Proposition~\ref{prop:comp-ges}, provides a more conservative but
computationally feasible test for boundedness.

\begin{corollary}\longthmtitle{Boundedness}\label{cor:bdd}
  Consider the network dynamics~\eqref{eq:dyn} and assume that
  $\dbf(t)$ is bounded, i.e., there exists $\bar \dbf \in \realpos^n$ such
  that $\dbf(t) \le \bar \dbf, t \ge 0$. If $\rho([\Wbf]_\zeros^\infty) < 1$, then the
  network trajectories remain bounded for all $t \ge 0$.
\end{corollary}
\begin{proof}
  If $\dbf(t)$ is constant, the result follows from
  Theorem~\ref{thm:bdd} and Proposition~\ref{prop:comp-ges}. 
  If $\dbf(t)$ is not constant, the same argument proves boundedness
  of trajectories for
  \begin{align}\label{eq:x-bar-d-bar}
    \tau \dot{\bar \xbf}(t) = -\bar \xbf(t) + [[\Wbf]_\zeros^\infty \bar \xbf(t) + \bar \dbf]_\zeros^\infty,
    \quad \bar \xbf(0) = \xbf_0.
  \end{align}
  The result then follows from the quasi-monotonicity of $(\xbf, t)
  \mapsto -\xbf + [\Wbf^+ \xbf + \bar \dbf]_\zeros^\infty$, similar to
  Lemma~\ref{lem:inh-bdd}.
\end{proof}

\begin{example}\longthmtitle{Uniform excitatory-inhibitory networks,
    cont'd}\label{ex:ei3}
  {\rm Let us revisit the excitatory-inhibitory network of
    Example~\ref{ex:ei}, here with $\mbf = \infty \ones_2$.  Clearly,
    the excitatory-only dynamics have bounded trajectories if and only
    if
  \begin{align}\label{eq:ei-bdd}
    \rho([\Wbf_{EI}]_\zeros^\infty) < 1 \Leftrightarrow \alpha n w_{ee} < 1,
  \end{align}
  which is the same condition as in~\eqref{eq:ei-h}
  and~\eqref{eq:wceue}. However, an exhaustive inspection of the
  switching regions $\{\Omega_\sigmab\}_\sigmab$ and the eigenvalues
  of $\{-\Ibf + \Sigmab^\ell \Wbf\}_\sigmab$ reveals that the
  boundedness of trajectories can also be guaranteed with the weaker
  condition
\begin{align*}
    &-\Ibf + \Wbf \ \text{be Hurwitz} \\
    &\Leftrightarrow \begin{cases}
      (1 - \alpha n w_{ee}) + (1 - (1 - \alpha) n w_{ii}) > 0, \ \text{and} \\
      (1 - \alpha n w_{ee}) (1 - (1 - \alpha) n w_{ii}) > \alpha (1 -
      \alpha) n^2 w_{ie} w_{ei},
    \end{cases}
  \end{align*}
  showing that there is room for sharpening Theorem~\ref{thm:bdd}.
  \oprocend }
\end{example}

\begin{remark}\longthmtitle{Comparison with the literature}
  {\rm In this section, we have provided a comprehensive list of
    conditions that both extend and simplify the state of the art on
    stability of dynamically isolated linear-threshold networks.  To
    the best of our knowledge, all the results are novel for the
    bounded case with the exception of the sufficiency of $\Ibf - \Wbf
    \in \Pc$ for the uniqueness of equilibria (one of the four
    implications in Theorem~\ref{thm:eue}) shown in~\cite{MF-AT:95}
    and Theorem~\ref{thm:ges}(i), for which we present a simpler
    proof.  Regarding unbounded networks, for equilibria we have
    extended~\cite[Thm 5.3]{DK-RL:87} (implying only the sufficiency
    of $\Ibf - \Wbf \in \Pc$ for EUE) to show both necessity and
    sufficiency of $\Ibf - \Wbf \in \Pc$ for both existence and
    uniqueness (Theorem~\ref{thm:eue}) and provided several results
    that partially characterize equilibria when $\Ibf - \Wbf \notin
    \Pc$.  On exponential stability of the unbounded case,
    Theorem~\ref{thm:ges} gives a simpler proof than~\cite[Thm
    1]{AP-NvdW-HN:05} for the sufficiency of $\Wbf \in \Lc$ and a
    novel proof for the necessity of $-\Ibf + \Wbf \in \Hc$.  Finally,
    our result on boundedness of trajectories (Theorem~\ref{thm:bdd})
    extends Corollary~\ref{cor:bdd} (also available in~\cite[Thm
    1]{HW-WB-HR:01}) to a much wider class of networks by exploiting
    the quasi-monotonicity of excitatory-only dynamics.  \oprocend}
\end{remark}

\section{Selective Inhibition in Bilayer Networks}\label{sec:stabilization}

Here, we study selective inhibition in bilayer networks as a building
block towards the understanding of hierarchical selective recruitment
in multilayer networks. With respect to the model described in
Section~\ref{sec:prob-form}, we consider two layers ($N = 2$), where
the dynamics of the lower layer $\Nc_2$ is described by~\eqref{eq:dyn}
and the dynamics of the upper layer $\Nc_1$ is arbitrary (this is for
generality, we consider linear-threshold dynamics for $\Nc_1$ too in a
multilayer framework in our accompanying work~\cite{EN-JC:21-tacII}).
Our goal is to study the selective inhibition of $\Nc_2^\ssm$ via the
input that it receives from~$\Nc_1$.

As pointed out in Section~\ref{sec:prob-form}, when a group of neurons
are inhibited, their activity is substantially decreased, ideally such
that their net input (their respective component of $\Wbf \xbf(t) +
\dbf(t)$) becomes negative and their firing rate decays exponentially
to zero. Therefore, the problem of selective inhibition is equivalent
to the exponential stabilization of the nodes $\Nc_2^\ssm$ to the
origin.  To this end, we decompose $\dbf(t)$ as
\begin{align}\label{eq:d-decomp}
  \dbf(t) = \Bbf \ubf(t) + \tilde \dbf.
\end{align}
The role of $\ubf(t) \in \realnonneg^p$ is to stabilize~$\Nc_2^\ssm$
to the origin while the role of $\tilde \dbf \in \real^n$ is to shape
the activity of $\Nc_2^\ssp$ by determining its equilibrium.  For the
purpose of this section, we assume $\tilde \dbf$ is given and
constant.

Let $r \le n$ be the size of $\Nc_2^\ssm$.  We partition $\xbf$,
$\Wbf$, $\Bbf$, and $\tilde \dbf$ accordingly,
\begin{align}\label{eq:wb}
  \xbf \!=\! \begin{bmatrix} \xbf^\ssm \\ \xbf^\ssp \end{bmatrix}\!, \, \Wbf \!=\!
  \begin{bmatrix}
    \Wbf^{\ssm\ssm} & \Wbf^{\ssm\ssp}
    \\
    \Wbf^{\ssp\ssm} & \Wbf^{\ssp\ssp}
  \end{bmatrix} \!, \,
   \Bbf
  \!=\! 
  \begin{bmatrix}
    \Bbf^\ssm
    \\
    \zeros
  \end{bmatrix}\!, \, \tilde \dbf \!=\! \begin{bmatrix} \zeros \\ \tilde \dbf^\ssp \end{bmatrix}\!,
\end{align}
where $\Wbf^{\ssm\ssm} \in \real^{r \times r}, \Bbf^\ssm \in
\realnonpos^{r \times p}$ is nonpositive to deliver inhibition, and
$\tilde \dbf^\ssp \in \real^{n - r}$.  $\mbf$ is decomposed similarly.
The first $r$ rows of $\Bbf$ are nonzero to allow for the inhibition
of $\Nc_2^\ssm$ while the remaining $n - r$ rows are zero to make this
inhibition selective to $\Nc_2^\ssm$.%
\footnote{This sparsity pattern can always be achieved by
  (re-)labeling the $r$ directly controlled nodes as $1, \dots, r$, so
  that the $n - r$ last entries of $\Bbf$ are $0$.}  The sparsity of
the entries of $\tilde \dbf$ is opposite to the rows of $\Bbf$ due to
the complementary roles of $\Bbf \ubf(t)$ and $\tilde \dbf$.
 
The mechanisms of inhibition in the brain are broadly
divided~\cite{JSI-MS:11} into two categories, feedforward and
feedback, based on how the signal $\ubf(t)$ is determined.  In the
following, we separately study each scenario, analyzing the interplay
between the corresponding mechanism and network structure.  We will
later combine both mechanisms when we discuss the complete HSR
framework in~\cite{EN-JC:21-tacII}, as natural selective inhibition is
not purely feedback or feedforward.

\subsection{Feedforward Selective Inhibition}

Feedforward inhibition~\cite{JSI-MS:11} refers to the scenario where
$\Nc_1$ provides an inhibition based on its own ``desired''
activity/inactivity pattern for $\Nc_2$ and irrespective of the
current state of $\Nc_2$. This is indeed possible if the inhibition is
sufficiently strong, as excessive inhibition has no effect on nodal
dynamics due to the (negative) thresholding in
$[\,\cdot\,]_0^m$. However, this independence from the activity level
of $\Nc_2$ requires some form of guaranteed boundedness, as defined
next.

\begin{definition}\longthmtitle{Monotone boundedness}
  The dynamics~\eqref{eq:dyn} is \emph{monotonically bounded} if for
  any $\bar \dbf \in \real^n$ there exists $\nub(\bar \dbf) \in
  \realnonneg^n$ such that $\xbf(t) \le \nub(\bar \dbf), t \ge 0$ for
  any $\dbf(t) \le \bar \dbf, t \ge 0$.
\end{definition}

From Lemma~\ref{lem:inh-bdd} and Proposition~\ref{prop:comp-ges},
\eqref{eq:dyn} is monotonically bounded if $\rho([\Wbf]_\zeros^\infty)
< 1$ and the initial condition $\xbf_0$ is restricted to a bounded
domain. Also in reality, the state of any biological neuronal network
is uniformly bounded due to the refractory period of its neurons,
implying monotone boundedness.  We next show that the GES of
$\Nc_2^\ssp$ is both necessary and sufficient for feedforward
selective inhibition.

\begin{theorem}\longthmtitle{Feedforward selective
    inhibition}\label{thm:dec-sub2}
  Consider the dynamics~\eqref{eq:dyn}, where the external input is
  given by~\eqref{eq:d-decomp}-\eqref{eq:wb} with a constant
  feedforward control
  \begin{align*}
    \ubf(t) \equiv \ubf \ge \zeros.
  \end{align*}
  Assume that \eqref{eq:dyn} is monotonically bounded and
  \begin{align}\label{eq:range}
    \range([\Wbf^{\ssm\ssm} \ \Wbf^{\ssm\ssp}]) \subseteq \range(\Bbf^\ssm).
  \end{align}
  Then, for any $\tilde \dbf^\ssp \in \real^{n - r}$, there exists
  $\bar \ubf \in \realnonneg^p$ such that for all $\ubf \ge \bar
  \ubf$, $\Nc_2$ is GES relative to a unique equilibrium of the form
  $\xbf_* = [\zeros_r^T \ (\xbf_*^\ssp)^{T}]^T$ if and only if
  $\Wbf^{\ssp\ssp}$ is such that the internal $\Nc_2^\ssp$ dynamics
  \begin{align}\label{eq:subnet}
    \tau \dot \xbf^\ssp = -\xbf^\ssp + [\Wbf^{\ssp\ssp} \xbf^\ssp +
    \tilde \dbf^\ssp]_\zeros^{\mbf\ssp}, 
  \end{align}
  is GES relative to a unique equilibrium.
\end{theorem}
\begin{proof} 
  ($\Leftarrow$) 
  Define $\ubf_s$ to be a solution of
  \begin{align}\label{eq:u0}
    \Bbf^\ssm \ubf_s = - [[\Wbf^{\ssm\ssm} \
    \Wbf^{\ssm\ssp}]]_\zeros^\infty \nub(\tilde \dbf).
  \end{align}
  This solution exists by assumption~\eqref{eq:range}. Let $\bar \ubf
  = [\ubf_s]_\zeros^\infty$ and note that $\Bbf^\ssm \ubf \le \Bbf^\ssm \bar \ubf \le
  \Bbf^\ssm \ubf_s$.  By construction,
  \eqref{eq:dyn},~\eqref{eq:d-decomp},~\eqref{eq:wb},~\eqref{eq:u0}
  simplify to
  \begin{align}\label{eq:dyn-d-sim}
    \notag \tau \dot \xbf^\ssm &= -\xbf^\ssm,
    \\
    \tau \dot \xbf^\ssp &= -\xbf^\ssp + [\Wbf^{\ssp\ssm} \xbf^\ssm +
    \Wbf^{\ssp\ssp} \xbf^\ssp + \tilde \dbf^\ssp]_\zeros^{\mbf\ssp},
  \end{align} 
  whose GES follows from Lemma~\ref{lem:cas}.
  
  ($\Rightarrow$) By monotone boundedness and nonpositivity of
  $\Bbf^\ssm$, $\xbf(t) \le \nu(\tilde \dbf)$ for all $t \ge 0$ and any
  $\ubf \ge \bar \ubf$. Let $\ubf = \bar \ubf + [\ubf_s]_\zeros^\infty$ where
  $\ubf_s$ is a solution to~\eqref{eq:u0}. Similar to above, this
  simplifies
  \eqref{eq:dyn},~\eqref{eq:d-decomp},~\eqref{eq:wb},~\eqref{eq:u0}
  to~\eqref{eq:dyn-d-sim},
  which is GES by assumption. However, for any initial condition of
  the form $\xbf(0) = [\zeros_r^T \ \xbf^\ssp(0)^T]^T$, the trajectories
  of~\eqref{eq:dyn-d-sim} are the same as~\eqref{eq:subnet}, and the
  result follows.
\end{proof} 

The next section shows that the condition~\eqref{eq:range} on the
ability to influence the dynamics of the task-irrelevant nodes through
control also plays a key role in feedback selective inhibition. We
defer the discussion about the interpretation of this condition to
Section~\ref{subsec:network-size} below.

\subsection{Feedback Selective Inhibition}

The core idea of feedback inhibition~\cite{JSI-MS:11}, as found
throughout the brain, is the dependence of the amount of inhibition on
the activity level of the nodes that are to be inhibited. This
dependence is in particular relevant to GDSA, as the stronger and more
salient a source of distraction, the harder one must try to suppress
its effects on perception.  The next result provides a novel
characterization of several equivalences between the dynamical
properties of $\Nc_2$ under linear full-state feedback inhibition and
those of $\Nc_2^\ssp$.

\begin{theorem}\longthmtitle{Feedback selective inhibition}\label{thm:dec-sub}
  Consider the dynamics~\eqref{eq:dyn}, where the external input is
  given by~\eqref{eq:d-decomp}-\eqref{eq:wb} with a linear state
  feedback $\ubf$
  \begin{align}\label{eq:ukx}
    \ubf(t) = \Kbf \xbf(t),
  \end{align}
  and $\Kbf \in \real^{p \times n}$ is a constant control gain. Assume
  that~\eqref{eq:range} holds.  Then, there almost always exists $\Kbf \in \real^{p
    \times n}$ such that
  \begin{enumerate}
  \item $\Ibf - (\Wbf + \Bbf \Kbf) \in \Pc$ if and only if $\Ibf -
    \Wbf^{\ssp\ssp} \in \Pc$;
  \item $-\Ibf + (\Wbf + \Bbf \Kbf) \in \Hc$ if and only if $-\Ibf +
    \Wbf^{\ssp\ssp} \in \Hc$;
  \item $\Wbf + \Bbf \Kbf \in \Lc$ if and only if $\Wbf^{\ssp\ssp} \in \Lc$;
  \item $\rho(|\Wbf + \Bbf \Kbf|) < 1$ if and only if
    $\rho(|\Wbf^{\ssp\ssp}|) < 1$;
  \item $\|\Wbf + \Bbf \Kbf\| < 1$ if and only if $\|[\Wbf^{\ssp\ssm} \
    \Wbf^{\ssp\ssp}]\| < 1$. 
  \end{enumerate} 
\end{theorem}
\begin{proof}
  \emph{(i)} $\Rightarrow$) For any $\Kbf = [\Kbf^\ssm \ \Kbf^\ssp] \in
  \real^{p \times n}$,
  \begin{align}\label{eq:wbk}
    \Wbf + \Bbf \Kbf =
    \begin{bmatrix} 
      \Wbf^{\ssm\ssm} + \Bbf^\ssm \Kbf^\ssm & \Wbf^{\ssm\ssp} + \Bbf^\ssm \Kbf^\ssp
      \\
      \Wbf^{\ssp\ssm} & \Wbf^{\ssp\ssp} 
    \end{bmatrix}.
  \end{align}
  Thus, since any principal submatrix of a P-matrix is a P-matrix,
  $\Ibf - \Wbf^{\ssp\ssp} \in \Pc$.
  \\
  $\Leftarrow$) By~\eqref{eq:range} there exists $\bar \Kbf \in
  \real^{p \times n}$ such that
  \begin{align}\label{eq:Kbar}
    -\begin{bmatrix} \Wbf^{\ssm\ssm} & \Wbf^{\ssm\ssp} \end{bmatrix} = \Bbf^\ssm \bar
    \Kbf.
  \end{align}
  Using the fact that the determinant of any block-triangular matrix
  is the product of the determinants of the blocks on its
  diagonal~\cite[Prop 2.8.1]{DSB:09},
  it follows that
    $\Ibf - (\Wbf + \Bbf \bar \Kbf) \in \Pc$.
  
  \emph{(ii)} $\Rightarrow$) This follows from~\eqref{eq:wbk}
  and the fact that a principal submatrix of a totally-Hurwitz
  matrix is totally-Hurwitz.
  \\
  $\Leftarrow$) Using the matrix $\bar \Kbf$
  in~\eqref{eq:Kbar}, the result follows from 
  the fact that the eigenvalues of a block-triangular matrix are the
  eigenvalues of its diagonal blocks.
    
  \emph{(iii)} $\Rightarrow$) Let $\Pbf = \Pbf^T > \zeros$ be such that
  \begin{align}\label{eq:-Iwbk}
    \!\!\!\! (-\Ibf + (\Wbf \!+\! \Bbf \Kbf)^T \Sigmab^\ell) \Pbf \!+\! \Pbf (-\Ibf +
    \Sigmab^\ell (\Wbf \!+\! \Bbf \Kbf)) < \zeros
  \end{align}
  for all $\sigmab \in \{0, \ell\}^n$. Consider, in particular,
  $\sigmab = [\zeros_r^T \ (\sigmab^\ssp)^T]^T$ where $\sigmab^\ssp
  \in \{0, \ell\}^{n - r}$ is arbitrary. Let $\Sigmab^{\ell\ssp} \in
  \real^{(n-r)\times(n-r)}$ be the bottom-right block of $\Sigma^\ell$
  and partition $\Pbf$ in $2$-by-$2$ block form similarly to
  $\Wbf$. Since
  \begin{align*}
    (-\Ibf &+ (\Wbf + \Bbf \Kbf)^T \Sigmab^\ell) \Pbf + \Pbf (-\Ibf +
    \Sigmab^\ell (\Wbf + \Bbf \Kbf))
    \\
    &= 
    \begin{bmatrix}
      \star & \star
      \\
      \star & (-\Ibf + \Sigmab^{\ell\ssp} \Wbf^{\ssp\ssp})^T \Pbf^{\ssp\ssp} + \Pbf^{\ssp\ssp}
      (-\Ibf + \Sigmab^{\ell\ssp} \Wbf^{\ssp\ssp})
    \end{bmatrix},
  \end{align*}
  and any principal submatrix of a negative definite matrix is
  negative definite, we deduce $\Wbf^{\ssp\ssp} \in \Lc$.
  \\
  $\Leftarrow$) Let $\Pbf^{\ssp\ssp} \in \real^{(n - r) \times (n - r)}$
  be such that 
  \begin{align*}
  (-\Ibf + (\Wbf^{\ssp\ssp})^T \Sigmab^{\ell\ssp}) \Pbf^{\ssp\ssp} +
  \Pbf^{\ssp\ssp} (-\Ibf + \Sigmab^{\ell\ssp} \Wbf^{\ssp\ssp}) < \zeros,
  \end{align*}
  for all $\sigmab^\ssp
  \in \{0, 1\}^{n - r}$ and $\bar \Kbf$ be as
  in~\eqref{eq:Kbar}. For any $\sigmab = [(\sigmab^\ssm)^T \
  (\sigmab^\ssp)^T]^T$, 
  \eqref{eq:Kbar} gives
  \begin{align*}
    -\Ibf + \Sigmab^\ell (\Wbf + \Bbf \bar \Kbf) 
    =
    \begin{bmatrix}
      -\Ibf & \zeros
      \\
      \star & -\Ibf + \Sigmab^{\ell\ssp} \Wbf^{\ssp\ssp} 
    \end{bmatrix}.
  \end{align*}
  Thus, the dynamics $\tau \dot \xbf = \big(-\Ibf + \Sigmab^\ell (\Wbf
  + \Bbf \bar \Kbf)\big) \xbf$ is a cascade of $\tau \dot \xbf^\ssm =
  -\xbf^\ssm$ and $\tau \dot \xbf^\ssp = (-\Ibf + \Sigmab^{\ell\ssp}
  \Wbf^{\ssp\ssp}) \xbf^\ssp + \star \cdot \xbf^\ssm$, where the
  latter has the ISS\footnote{Input-to-state stability}-Lyapunov
  function $V^\ssp(\xbf^\ssp) = (\xbf^\ssp)^T \Pbf^{\ssp\ssp}
  \xbf^\ssp$. Using~\cite[Thm 3]{EDS:95a}, \eqref{eq:-Iwbk} holds for
  $\Kbf = \bar \Kbf$, $\Pbf = \diag(\Ibf, \Pbf^{\ssp\ssp})$, and any
  $\sigmab \in \{0, \ell\}^n$, giving $\Wbf + \Bbf \bar \Kbf \in \Lc$.
  
  \emph{(iv)} $\Rightarrow$) This follows from~\eqref{eq:wbk}
  and~\cite[Fact 4.11.19]{DSB:09}.
  \\
  $\Leftarrow$) Consider the matrix $\bar \Kbf$
  in~\eqref{eq:Kbar}. The result then follows from the fact that the
  eigenvalues of a block-triangular matrix are the eigenvalues of its
  diagonal blocks.
  
  \emph{(v)} $\Rightarrow$) Note that for any $\Kbf \in \real^{p
    \times n}$,
  \begin{align*}
    \|&\Wbf + \Bbf \Kbf\|^2 = \rho\left(
      \begin{bmatrix}
        \star & \star
        \\
        \star & \Wbf^{\ssp\ssm} (\Wbf^{\ssp\ssm})^T + \Wbf^{\ssp\ssp} (\Wbf^{\ssp\ssp})^T
      \end{bmatrix}
    \right)
    \\
    &\ge \rho(\Wbf^{\ssp\ssm} (\Wbf^{\ssp\ssm})^T + \Wbf^{\ssp\ssp} (\Wbf^{\ssp\ssp})^T) =
    \left\|
      \begin{bmatrix}
        \Wbf^{\ssp\ssm} & \Wbf^{\ssp\ssp}
      \end{bmatrix}
    \right\|^2,
  \end{align*}
  where the inequality follows from the well-known interlacing
  property of eigenvalues of principal submatrices
  (cf.~\cite{CRJ-HAR:81}).
  \\
  $\Leftarrow$) Consider the matrix $\bar \Kbf$ in~\eqref{eq:Kbar} and
  note that
  \begin{align*}
    &\|\Wbf + \Bbf \bar \Kbf\|^2 = \rho\left(
      \begin{bmatrix}
        \zeros & \zeros 
        \\
        \zeros & \Wbf^{\ssp\ssm} (\Wbf^{\ssp\ssm})^T + \Wbf^{\ssp\ssp}
        (\Wbf^{\ssp\ssp})^T 
      \end{bmatrix}
    \right)
    \\
    &= \rho(\Wbf^{\ssp\ssm} (\Wbf^{\ssp\ssm})^T \!+\! \Wbf^{\ssp\ssp}
    (\Wbf^{\ssp\ssp})^T) = \left\|
      \begin{bmatrix}
        \Wbf^{\ssp\ssm} &
        \!\!\! \Wbf^{\ssp\ssp}
      \end{bmatrix}
    \right\|^2 < 1,
  \end{align*}
  completing the proof.
\end{proof}

\begin{remark}\longthmtitle{Feedback inhibition with nonnegative
    $\ubf(t)$}\label{rem:nnegfb}
  {\rm Even though Theorem~\ref{thm:dec-sub} is motivated by feedback
    inhibition in the brain, the result illustrates some fundamental
    properties of linear-threshold dynamics and the corresponding
    matrix classes that is of independent interest, which motivates
    the generality of its formulation. The particular application to
    brain networks requires nonnegative inputs, which we discuss next.
    The core principle of Theorem~\ref{thm:dec-sub} is the
    cancellation of local input $[\Wbf^{\ssm\ssm} \ \Wbf^{\ssm\ssp}] \xbf$ to
    $\Nc_2^\ssm$ with the top-down feedback input $\Bbf^\ssm \bar \Kbf
    \xbf$, simplifying the dynamics of $\Nc_2^\ssm$ to $\tau \dot \xbf^\ssm
    = - \xbf^\ssm$ that guarantee its inhibition. However, the resulting
    input signal $\ubf = \bar \Kbf \xbf$ (being the firing rate of
    some neuronal population in $\Nc_1$) may not remain nonnegative at all
    times. This can be easily addressed as follows. Similar to the
    proof of Theorem~\ref{thm:dec-sub2}, we let
    \begin{align*}
      \ubf(t) = [\bar \Kbf \xbf(t)]_\zeros^\infty.
    \end{align*}
    This makes $\ubf(t)$ nonnegative without affecting the selective
    inhibition of $\Nc_2^\ssm$ in~\eqref{eq:dyn} as $\Bbf^\ssm \le
    \zeros$.  In principle, a similar concern can exist as to whether
    $\bar \Kbf \xbf$ becomes larger than the maximum firing rate of
    the corresponding populations in $\Nc_1$. However, this only
    relates to the magnitude of the entries in $\Bbf^\ssm$
    (via~\eqref{eq:Kbar}, as opposed to the sign of $\bar \Kbf \xbf$,
    which relates to the sign of the entries in $\Bbf^\ssm$), which
    can always be increased via synaptic long term potentiation
    (LTP)~\cite{RAN:17}, in turn decreasing the magnitude of the
    entries in $\bar \Kbf$.\oprocend}
\end{remark}

\begin{remark}\longthmtitle{State vs. output feedback}
    \rm The assumption of state feedback is a simplifying one and its
    generalization merits further research. However, we note that
    $\xbf^\ssm$ is most likely always available for feedback (as
    feedback inhibition is highly reciprocal at the neuronal
    level~\cite{RT-SL-BR:16} and even more so at the population level)
    while the availability of $\xbf^\ssp$ for feedback remains
    case-specific. If the latter is not available, one of two
    scenarios may happen: either the local interaction of $\xbf^\ssm$
    and $\xbf^\ssp$ is competitive and $\Wbf^{\ssm\ssp} \le \zeros$
    (which is not unlikely due to the prevalence of lateral inhibition
    in the cortex~\cite{PS-GT-RL-EHB:98}), in which case the feedback
    of $\xbf^\ssp$ is not even needed (similar to
    Remark~\ref{rem:nnegfb}) or, at worst, the unobserved $\xbf^\ssp$
    actively excites $\xbf^\ssm$, in which case a combination of
    feedback and feedforward inhibition can be used, similar to our
    full model in Part II~\cite[Thm 4.3]{EN-JC:21-tacII}. \oprocend
\end{remark}

\subsection{Network Size, Weight Distribution, and
  Stabilization}\label{subsec:network-size}

Underlying the discussion above is the requirement that $\Nc_2$ can be
asymptotically stabilized towards an equilibrium which has some
components equal to zero and the remaining components determined by
$\tilde \dbf$.  Here, it is important to distinguish between the
stability of $\Nc_2$ in the absence and presence of selective
inhibition. In reality, the large size of biological neuronal networks
often leads to highly unstable dynamics if all the nodes in a layer,
say $\Nc_2$, are active. Therefore, the selective inhibition of
$\Nc_2^\ssm$ is not only responsible for the suppression of the
task-irrelevant activity of $\Nc_2^\ssm$, but also for the overall
stabilization of $\Nc_2$ that allows for top-down recruitment
of~$\Nc_2^\ssp$.  This poses limitations on the size and structure of the
subnetworks $\Nc_2^\ssm$ and $\Nc_2^\ssp$.  It is in this context that one
can analyze the condition~\eqref{eq:range} assumed in both
Theorems~\ref{thm:dec-sub2} and~\ref{thm:dec-sub}.  This condition
requires, essentially, that there are sufficiently many
``independent'' external controls $\ubf$ to enforce inhibition on
$\Nc_2^\ssm$.  The following result formalizes this statement.

\begin{lemma}\longthmtitle{Equivalent characterization
    of~\eqref{eq:range}}\label{lem:range}
  Let the matrices $\Wbf^\ssm$ and $\Bbf^\ssm$ have dimensions $r
  \times n$ and $r \times p$, respectively. Then, $\range(\Wbf^\ssm)
  \subseteq \range(\Bbf^\ssm)$ for almost all $(\Wbf^\ssm, \Bbf^\ssm)
  \in \real^{r \times n} \times \real^{r \times m}$ if and only if $p
  \ge r$.
\end{lemma}
\begin{proof} 
  $\Rightarrow$) Assume, by contradiction, that $p < r$, so
  $\range(\Bbf^\ssm) \subsetneq \real^r$ for any $\Bbf^\ssm$. Let
  $\Qbf = \Qbf(\Bbf^\ssm)$ be a matrix whose columns form a basis for
  $\range(\Bbf^\ssm)^\perp$. Then, $\range(\Wbf^\ssm) \subseteq
  \range(\Bbf^\ssm)$ if and only if $\Qbf(\Bbf^\ssm)^T \Wbf^\ssm =
  \zeros$. By Fubini's theorem~\cite[Ch. 20]{HLR-PF:10},
  \begin{align*}
    &\int_{\real^{r \times n} \times \real^{r \times p}}
    \mathbb{I}_{\{\Qbf(\Bbf^\ssm)^T \Wbf^\ssm = \zeros\}}(\Wbf^\ssm, \Bbf^\ssm) d
    (\Wbf^\ssm, \Bbf^\ssm)
    \\
    &= \int_{\real^{r \times p}} d \Bbf^\ssm \int_{\real^{r \times n}}
    \mathbb{I}_{\{\Qbf(\Bbf^\ssm)^T \Wbf^\ssm = \zeros\}}(\Wbf^\ssm, \Bbf^\ssm) d
    \Wbf^\ssm
    \\
    &= \int_{\real^{r \times p}} 0 \ d \Bbf^\ssm = 0,
  \end{align*}
  where $\mathbb{I}$ denotes the indicator function. This
  contradiction proves $p \ge r$.
  \\
  $\Leftarrow$) Let $\Bbf^\ssm = [\Bbf_1^\ssm \ \Bbf_2^\ssm]$ where
  $\Bbf_1^\ssm \in \real^{r \times r}$. It is straightforward to show
  that
  \begin{align*}
    &\setdef{(\Wbf^\ssm, \Bbf^\ssm)}{\range(\Wbf^\ssm) \nsubseteq
      \range(\Bbf^\ssm)} \subseteq \real^{r \times n} \times A,
  \end{align*}
  where $A = \setdef{\Bbf^\ssm}{\det(\Bbf_1^\ssm) = 0}$. Since $A$ has
  measure zero, the result follows from a similar argument as above
  invoking Fubini's theorem.
\end{proof}

Based on intuitions from linear systems theory, it may be tempting to
seek a relaxation of~\eqref{eq:range} for the case where $p < r$. This
is due to the fact that for a \emph{linear} system $\tau \dot \xbf =
\Wbf \xbf + \Bbf \ubf$, it is known~\cite[eq (4.5) and Thm
3.5]{CTC:98} that
the set of all reachable states from the origin is given by
\begin{align*}
  \range\left(
    \begin{bmatrix}
      \Bbf^\ssm & \Wbf^{\ssm\ssm} \Bbf^\ssm & \cdots & (\Wbf^{n - 1})^{\ssm\ssm} \Bbf^\ssm
      \\
      \zeros & \Wbf^{\ssp\ssm} \Bbf^\ssm & \cdots & (\Wbf^{n - 1})^{\ssp\ssm} \Bbf^\ssm
    \end{bmatrix}
  \right),
\end{align*}
which is usually much larger than $\range(\Bbf)$. Therefore, it is
reasonable to expect that~\eqref{eq:range} could be relaxed to
{\interdisplaylinepenalty=10000
\begin{align}\label{eq:range-rel}
  \notag \range([&\Wbf^{\ssm\ssm} \ \Wbf^{\ssm\ssp}]) 
  \\
  &\quad \subseteq \range([\Bbf^\ssm \ \Wbf^{\ssm\ssm}
  \Bbf^\ssm \ \cdots \ (\Wbf^{n - 1})^{\ssm\ssm} \Bbf^\ssm]).
\end{align}
}
However, it turns out that this relaxation is not possible, the reason
being the (apparently simple, yet intricate) nonlinearity
in~\eqref{eq:dyn}. We show this by means of an example.

\begin{example}\longthmtitle{Tightness of~\eqref{eq:range}}
  {\rm Consider the feedback
    dynamics~\eqref{eq:dyn},~\eqref{eq:d-decomp},~\eqref{eq:ukx},
    where $n=3$, $p=1$, $r = 2$, and
  \begin{align*}
    \Wbf = \left[\begin{array}{cc|c} 2\alpha & 0 & 0 \\ 0 & 3\alpha &
        0 \\ \hline 0 & 0 & \alpha \end{array}\right], \quad \Bbf =
    \left[\begin{array}{c} 1 \\ 1 \\ \hline 0 \end{array}\right],
    \quad \alpha \in (0.5, 1).
  \end{align*}
  Clearly, \eqref{eq:range} does not hold (so
  Theorem~\ref{thm:dec-sub}(iv) does not apply), but
  $\range({[\Wbf^{\ssm\ssm} \ \Wbf^{\ssm\ssp}]}) \subseteq \range([\Bbf^\ssm \
  \Wbf^{\ssm\ssm} \Bbf^\ssm])$. One can show that for all $\Kbf \in \real^{1
    \times 3}$,
  \begin{align*}
    \rho(|\Wbf + \Bbf \Kbf|) \ge 2 \alpha > 1,
  \end{align*}
  while $\rho(\Wbf^{\ssp\ssp}) = \alpha < 1$, verifying
  that~\eqref{eq:range} is necessary and cannot be relaxed
  to~\eqref{eq:range-rel}.  \oprocend }
\end{example}

Theorems~\ref{thm:dec-sub2} and~\ref{thm:dec-sub} use completely
different mechanisms for inhibition of $\Nc_2^\ssm$, yet they are
strikingly similar in one conclusion: that the dynamical
properties achievable under selective inhibition are precisely those
satisfied by 
$\Nc_2^\ssp$. This has important
implications for the size and structure of the part $\Nc_2^\ssp$ that can
be active at any instance of time without resulting in
instability. The next remark elaborates on this implication.

\begin{remark}\longthmtitle{Implications for the size and connection
    strength of $\Nc_2^\ssp$} {\rm Existing experimental evidence
    suggest that the synaptic weights $\Wbf$ in cortical networks are
    sparse, approximately follow a log-normal distribution, and have a
    pairwise connection probability that is independent of physical
    distance between neurons within short
    distances~\cite{SS-PJS-MR-SN-DBC:05}.  Based on simulations of
    matrices with such statistics, Figure~\ref{fig:rho}(a, b) show how
    quickly the network (representing $\Nc_2$ here) moves towards
    instability when its size grows. On the other hand, it is
    well-known that increasing $n$ (and thus the number of synaptic
    weights) increases network expressivity (i.e., capacity to
    reproduce complex trajectories). While determining the optimal
    size of a network that leads to the best tradeoff between
    stability and expressivity is beyond the scope of this work, our
    results suggest a critical role for selective inhibition in
    keeping only a limited number of nodes in $\Nc_2$ active at any
    given time while inhibiting others. In other words, while the
    overall size of subnetworks in a brain network (corresponding to,
    e.g., the number of neuronal populations with distinct preferred
    stimuli in a brain region) is inevitably large, selective
    inhibition offers a plausible explanation for the mechanism by
    which the brain keeps the number of \emph{active populations}
    bounded ($O(1)$) at any given time.
    
    Similarly, Figure~\ref{fig:rho}(c, d) show the transition of
    networks towards instability as their synaptic connections become
    stronger. While excitatory synapses, as expected, have a larger
    impact on stability, the same trend is also observed while varying
    inhibitory synaptic strengths. Interestingly, several works in the
    neuroscience literature have shown that neuronal networks maintain
    stability by re-scaling their synaptic weights that change during
    learning, a process commonly referred to as \emph{homeostatic
      synaptic plasticity}~\cite{GT:12}. Our results thus open the way
    to provide rigorous and quantifiable measures of the optimal size
    and weight distribution of subnetworks that may be active at any
    given time and the homeostatic mechanisms that maintain any
    desired level of stability and expressivity. \oprocend}
\end{remark}

\begin{figure*}
\centering \subfloat[]{\small
  \includegraphics[width=0.225\textwidth]{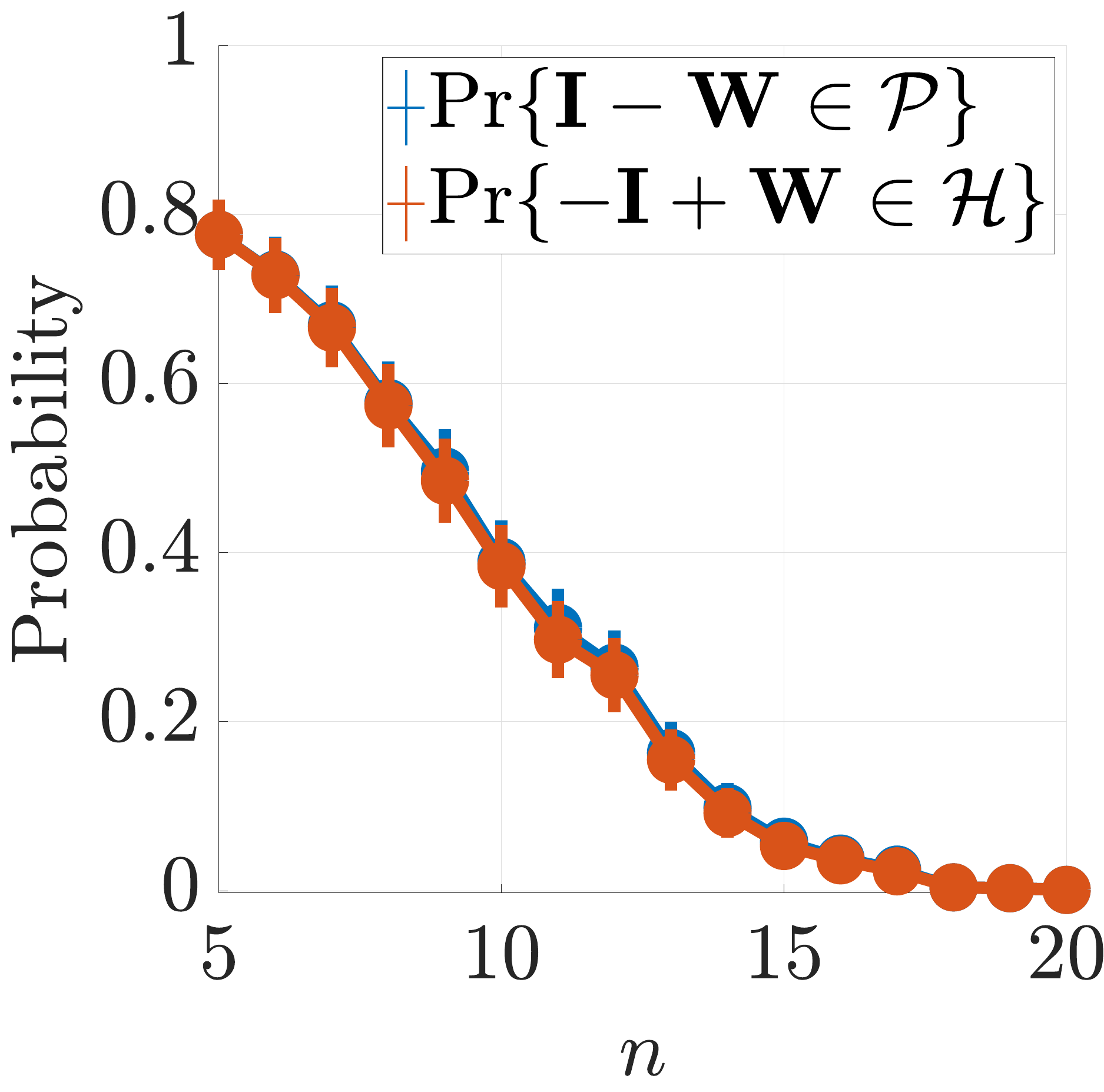}
  }
  \hspace{0.01\textwidth}
  \subfloat[]{\small
 \includegraphics[width=0.24\textwidth]{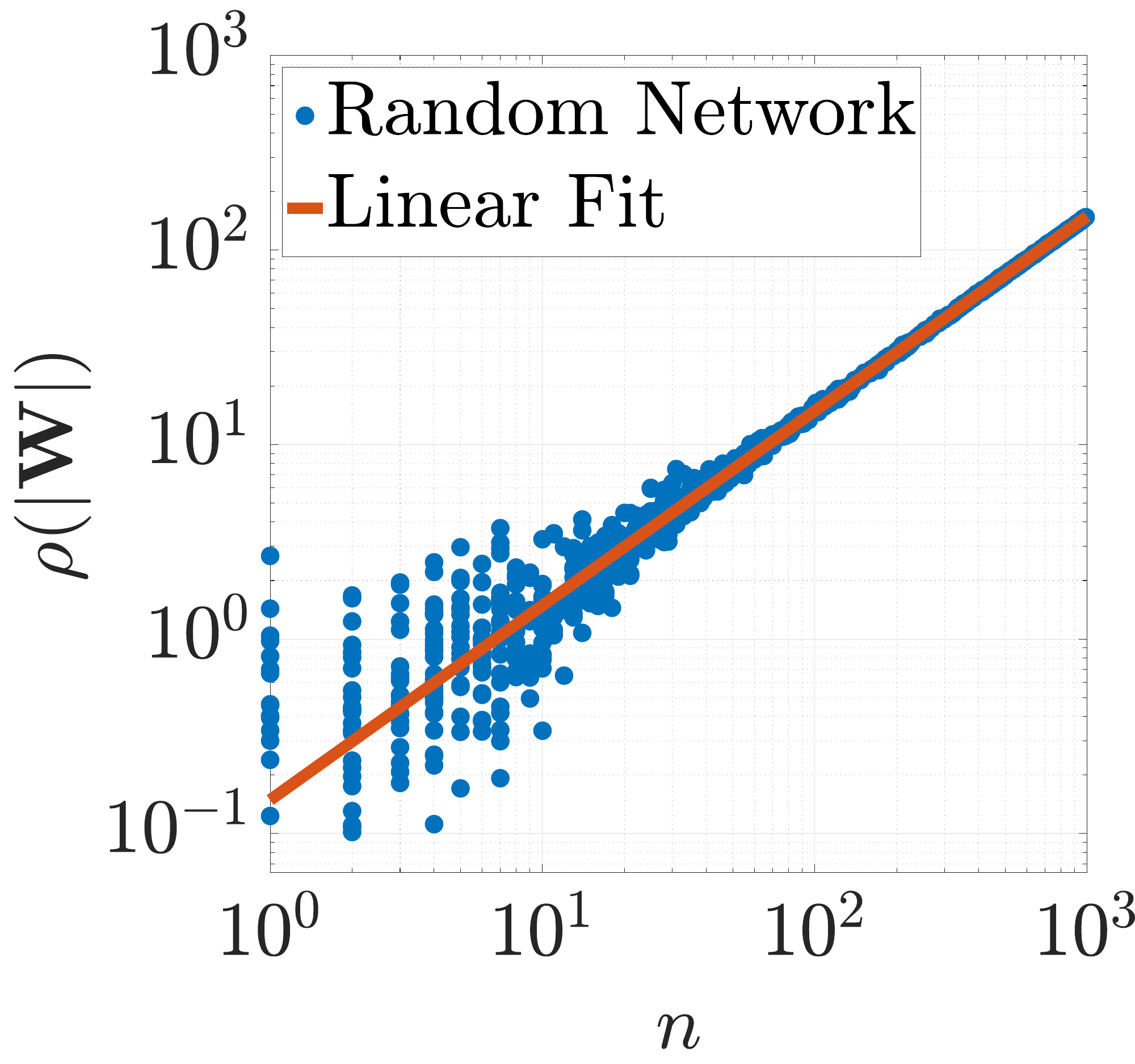}
  }
    \hspace{0.01\textwidth}
  \subfloat[]{\small
 \includegraphics[width=0.225\textwidth]{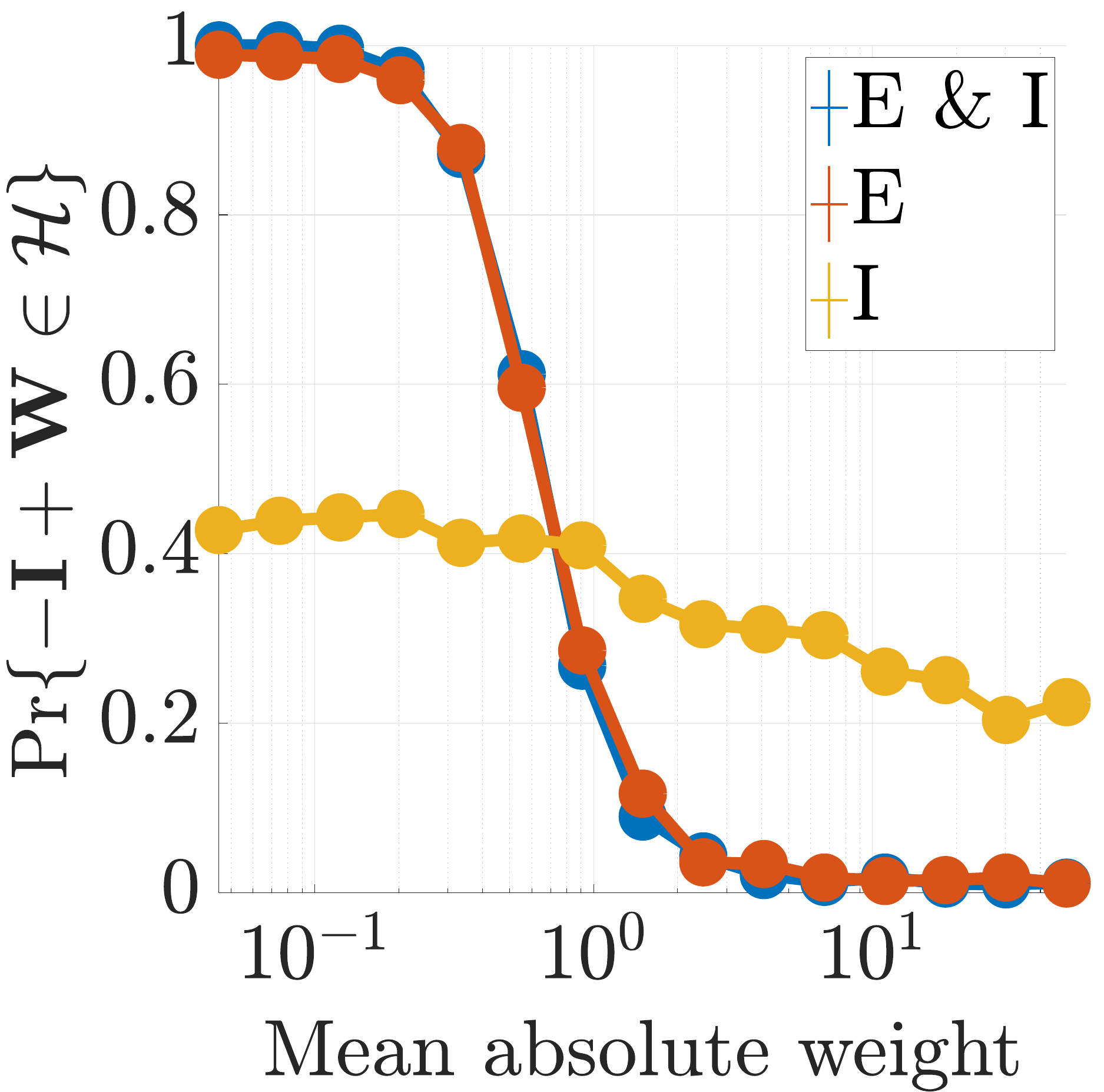}
  }
    \hspace{0.01\textwidth}
  \subfloat[]{\small
 \includegraphics[width=0.225\textwidth]{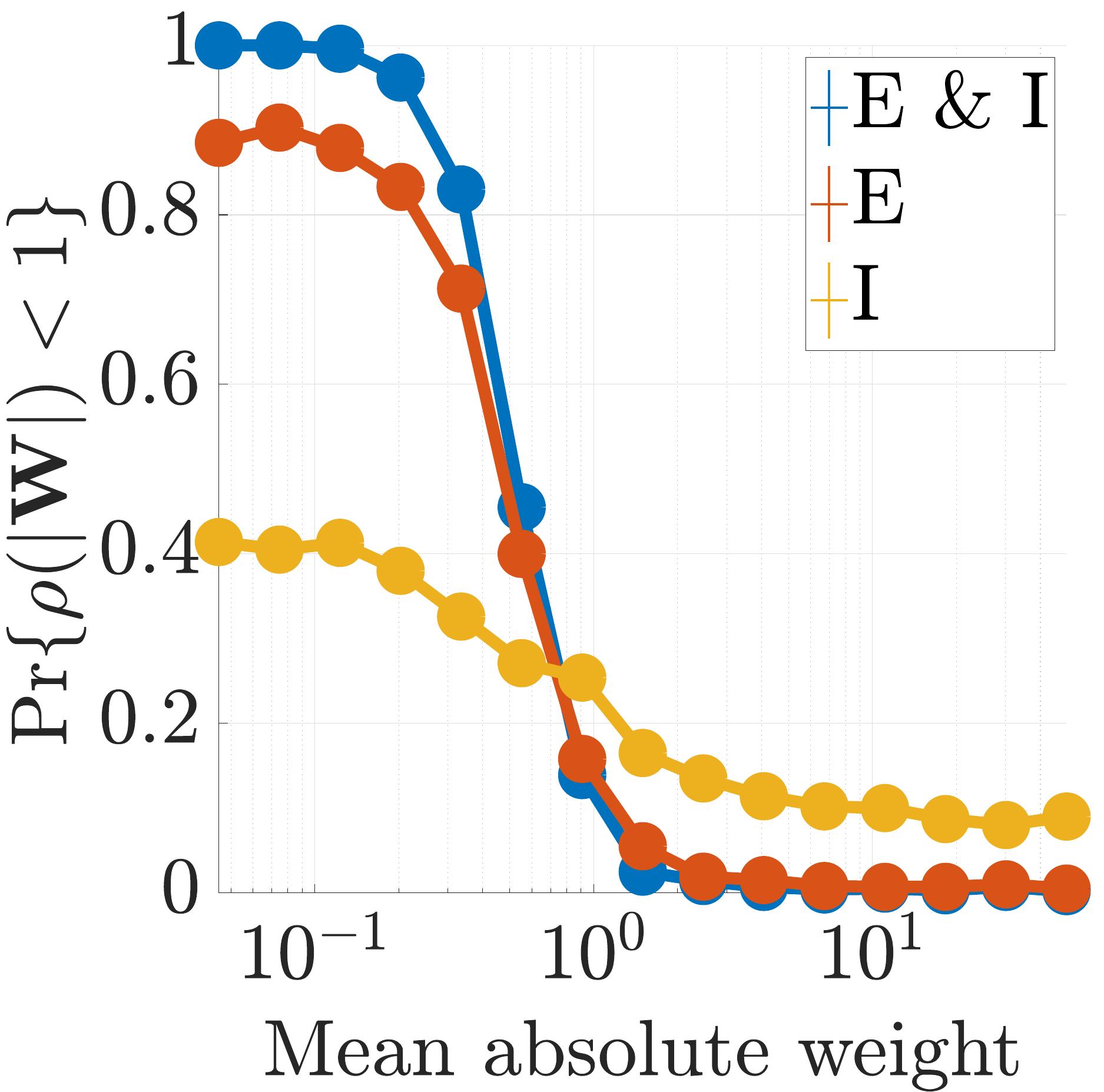}
  }
  \caption{The effects of network size and weight distribution on its
    stability. In all panels, unless otherwise stated, $\Wbf$ matrices
    are generated randomly with log-normally distributed entries with
    parameters $\mu = -0.7$ and $\sigma = 0.9$ as given
    in~\cite{SS-PJS-MR-SN-DBC:05}, $20\%$ sparsity, and $80\%$
    excitatory nodes. Probabilities are estimated empirically with
    $10^3$ samples and error bars represent standard error of the mean
    (s.e.m).  \textbf{(a)} Probability of $\Ibf - \Wbf \in \Pc$ (as
    related to EUE, cf. Theorem~\ref{thm:eue}) and $-\Ibf + \Wbf \in
    \Hc$ (as related to asymptotic stability,
    cf. Theorem~\ref{thm:ges}) showing a rapid decay with $n$.
    \textbf{(b)} For large $n$, while checking $\Ibf - \Wbf \in \Pc$
    and $-\Ibf + \Wbf \in \Hc$ can become prohibitive, their
    sufficient condition $\rho(|\Wbf|) < 1$ can be checked
    efficiently. The line illustrates a fit of the form $\log
    \rho(|\Wbf|) = \alpha \log n + \beta$ with $\alpha = 1$ and $\beta
    = -1.2$, showing a linear growth of $\rho(|\Wbf|)$ with $n$.
    \textbf{(c, d)} The probabilities of $-\Ibf + \Wbf \in \Hc$ and
    $\rho(|\Wbf|) < 1$ while varying the $\mu$ parameter of the
    log-normal weight distribution over $[-3.5, 3.5]$ for both
    excitatory and inhibitory synapses (blue), only for excitatory
    ones (red), or only for inhibitory ones (yellow). The plots are
    very similar for $\Ibf - \Wbf \in \Pc$ and thus not shown. $n =
    10$ is fixed in these panels.  }
  \label{fig:rho}
\vspace*{-10pt}
\end{figure*}

\section{Conclusions}\label{sec:conclusions}

We adopt a control-theoretic framework, termed hierarchical
  selective recruitment (HSR), as a mechanism to explain goal-driven
  selective attention. Motivated by the organization of the brain, HSR
  employs a hierarchical model which consists of an arbitrary number
  of neuronal subnetworks that operate at different layers of a
  hierarchy. While HSR is not confined to any family of models, we
here use the well-studied linear-threshold rate models to describe the
dynamics at each layer of the hierarchy.  We provide a thorough
analysis of the internal dynamics of each layer. Leveraging the
switched-affine nature of linear-threshold dynamics, we derive several
necessary and sufficient conditions for the existence and uniqueness
of equilibria (corresponding to P-matrices), local and global
asymptotic stability (corresponding to totally-Hurwitz matrices), and
boundedness of trajectories (corresponding to stability of
excitatory-only dynamics). These results set the basis for analyzing
the problem of selective inhibition.  We show that using either
feedforward or feedback inhibition, the dynamical properties of each
layer after inhibition are precisely determined by 
the task-relevant part that remains active. We have also provided
constructive control designs that guarantee selective inhibition under
both schemes.  Among the 
directions for future research, we highlight the study of
output-feedback 
selective inhibition and the analysis of the conditions on
(single-layer) linear-threshold networks that lead to the emergence of
limit cycles and their stability. 

\section*{Acknowledgments}
We would like to thank Dr. Erik J. Peterson for piquing our interest
with his questions on dimensionality control in brain networks and for
introducing us to linear-threshold modeling in neuroscience. We are
further indebted to the anonymous reviewer for suggesting the proof of
the necessity of $\Ibf - \Wbf \in \Pc$ for the EUE of unbounded
networks. This work was supported by NSF Award CMMI-1826065 (EN and
JC) and ARO Award W911NF-18-1-0213 (JC).

\setcounter{section}{0}
\renewcommand{\thesection}{Appendix \Alph{section}}
\section{Additional Lemmas and Proofs}\label{app:proofs}
\renewcommand{\thesection}{\Alph{section}}

\begin{proof}[Proof of Theorem~\ref{thm:eq-candid-P}]
  The necessity is trivial since $\Mbf_\zeros^{-1} \Mbf_\ellb = -(\Ibf
  - \Wbf)$. To prove sufficiency, note that for any $\sigmab \in \{0,
  \ell\}^n$,
  \begin{align}\label{eq:M-sigma-inv}
    \Mbf_\sigmab^{-1} = (\Ibf - \Wbf \Sigmab^\ell) (2 \Sigmab^\ell -
    \Ibf) = (2 \Sigmab^\ell - \Ibf) - \Wbf \Sigmab^\ell.
  \end{align}
  Since nodes can be relabeled arbitrarily, we can assume
  without loss of generality that $\sigmab_1 = [\ellb_{n_1}^T \ \
  \ellb_{n_2}^T \ \ \zeros_{n_3}^T \ \ \zeros_{n_4}^T]^T$ and
  $\sigmab_2 = [\ellb_{n_1}^T \ \ \zeros_{n_2}^T \ \ \ellb_{n_3}^T \ \
  \zeros_{n_4}^T]^T$ where $n_1, \ldots, n_4 \ge 0, \sum_{i = 1}^4 n_i
  = n$. Then, it follows from~\eqref{eq:M-sigma-inv} that
  \begin{align*}
    \Mbf_{\sigmab_1}^{-1} &= \left[\begin{array}{cc:cc} \Ibf_{n_1} - \Wbf_{11} &
        -\Wbf_{12}
        & \zeros & \zeros \\ -\Wbf_{21} & \Ibf_{n_2} - \Wbf_{22} & \zeros & \zeros \\ \hdashline
        -\Wbf_{31} & -\Wbf_{32} & -\Ibf_{n_3} & \zeros \\ -\Wbf_{41} & -\Wbf_{42} & \zeros
        & -\Ibf_{n_4} \end{array}\right],
    \\
    \Mbf_{\sigmab_2}^{-1} &= \left[\begin{array}{cccc} \Ibf_{n_1} - \Wbf_{11} & \zeros
        & -\Wbf_{13} & \zeros \\ -\Wbf_{21} & -\Ibf_{n_2} & -\Wbf_{23} & \zeros \\
        -\Wbf_{31} & \zeros & \Ibf_{n_3} - \Wbf_{33} & \zeros \\ -\Wbf_{41} & \zeros &
        -\Wbf_{43} & -\Ibf_{n_4} \end{array}\right],
  \end{align*}
  where $\Wbf_{ij}$'s are submatrices of $\Wbf$ with appropriate
  dimensions.  Taking the inverse of $\Mbf_{\sigmab_1}^{-1}$ as a
  $2$-by-$2$ block-triangular matrix~\cite[Prop 2.8.7]{DSB:09} (with the indicated blocks), we get
  {
  \addtolength\arraycolsep{-2pt}
  \begin{align*}
    \Mbf_{\sigmab_1} \!\!=\!\! \begin{bmatrix} \begin{bmatrix} \Ibf_{n_1} - \Wbf_{11} &
        -\Wbf_{12}
        \\ -\Wbf_{21} & \Ibf_{n_2} - \Wbf_{22} \end{bmatrix}^{-1} & \zeros \\
      -\!\!\begin{bmatrix} \Wbf_{31} & \Wbf_{32} \\ \Wbf_{41} &
        \Wbf_{42} \end{bmatrix} \!\!\! \begin{bmatrix} \Ibf_{n_1} \!\!-\!\! \Wbf_{11} & -\Wbf_{12} \\
        -\Wbf_{21} & \Ibf_{n_2} \!\!-\!\! \Wbf_{22} \end{bmatrix}^{-1} & -\Ibf_{n_3 + n_4} \end{bmatrix},
  \end{align*}
  }
  so direct multiplication gives
    $\Mbf_{\sigmab_1} \Mbf_{\sigmab_2}^{-1} =\begin{bmatrix} \Bbf_1 & \Bbf_2
        \\ \Bbf_3 & \Bbf_4 \end{bmatrix}$,
  with
  \begin{align*}
    \Bbf_1 &= \begin{bmatrix} \Ibf_{n_1} - \Wbf_{11} & -\Wbf_{12} \\ -\Wbf_{21} & \Ibf_{n_2}
        - \Wbf_{22} \end{bmatrix}^{-1} \begin{bmatrix} \Ibf_{n_1} -
        \Wbf_{11} & \zeros \\ -\Wbf_{21} & -\Ibf_{n_2} \end{bmatrix},
    \\
    \Bbf_2 &= -\begin{bmatrix} \Ibf_{n_1} - \Wbf_{11} & -\Wbf_{12} \\ -\Wbf_{21} &
        \Ibf_{n_2} - \Wbf_{22} \end{bmatrix}^{-1} \begin{bmatrix}
        \Wbf_{13} & \zeros \\ \Wbf_{23} & \zeros \end{bmatrix},
    \\
    \Bbf_3 &= -\begin{bmatrix} \Wbf_{31} & \Wbf_{32} \\ \Wbf_{41} &
        \Wbf_{42} \end{bmatrix} \Bbf_1 + \begin{bmatrix} \Wbf_{31}
        & \zeros \\ \Wbf_{41} & \zeros \end{bmatrix},
    \\
    \Bbf_4 &= -\begin{bmatrix} \Wbf_{31} & \Wbf_{32} \\ \Wbf_{41} &
        \Wbf_{42} \end{bmatrix} \Bbf_2 - \begin{bmatrix} \Ibf_{n_3} -
        \Wbf_{33} & \zeros \\ -\Wbf_{43} & -\Ibf_{n_4} \end{bmatrix}.
  \end{align*}
With this, after some computations one can show that 
  \begin{align}\label{eq:gamma-stars}
    \Mbf_{\sigmab_1} \Mbf_{\sigmab_2}^{-1} = \left[\begin{array}{ccc} \Ibf_{n_1}
        & \begin{array}{cc} \star & \star \end{array} & \zeros
        \\ \begin{array}{c} \zeros \\ \zeros \end{array} &
        \left[\begin{array}{cc}\star & \star \\ \star &
            \star \end{array}\right] & \begin{array}{c} \zeros \\
          \zeros \end{array} \\ \zeros & \begin{array}{cc} \star &
          \star \end{array} & \Ibf_{n_4} \end{array}\right].
  \end{align}
  Let $\Gammab \in
  \real^{(n_2 + n_3) \times (n_2 + n_3)}$ be the bracketed block in %
  $\Mbf_{\sigmab_1} \Mbf_{\sigmab_2}^{-1}$ and
  define
  \begin{align*}
    \Qbf &= \left[\begin{array}{cc} \Qbf_{11} & \Qbf_{12} \\ \Qbf_{21} &
        \Qbf_{22} \end{array}\right] \triangleq \left[\begin{array}{cc} \Ibf_{n_1}
        - \Wbf_{11} & -\Wbf_{12} \\ -\Wbf_{21} & \Ibf_{n_2} -
        \Wbf_{22} \end{array}\right]^{-1},
    \\
    \Rbf &\triangleq \Ibf_{n_3} - \Wbf_{33} - \left[\begin{array}{cc} \Wbf_{31} &
        \Wbf_{32} \end{array}\right] \Qbf \left[\begin{array}{c} \Wbf_{13} \\
        \Wbf_{23} \end{array}\right].
  \end{align*}
  It can be shown that
  \begin{align*}
    \Gammab = -\begin{bmatrix} \Qbf_{22} & \begin{bmatrix} \Qbf_{21} &
        \Qbf_{22} \end{bmatrix} \begin{bmatrix} \Wbf_{13} \\
        \Wbf_{23} \end{bmatrix} \\ -\begin{bmatrix} \Wbf_{31} &
        \Wbf_{32} \end{bmatrix} \begin{bmatrix} \Qbf_{12} \\
        \Qbf_{22} \end{bmatrix} & \Rbf \end{bmatrix}.
  \end{align*}
  Inverting the left-hand-side matrix (below) as a 2-by-2 block
  matrix~\cite[Prop 2.8.7]{DSB:09} (the first block is $\Qbf^{-1}$)
  and applying the matrix inversion lemma~\cite[Cor 2.8.8]{DSB:09} to
  the first block of the result, we obtain
  \begin{align*}
    \left[\!\!\begin{array}{cc:c}
      \Ibf_{n_1} - \Wbf_{11} & -\Wbf_{12} & -\Wbf_{13}
      \\ -\Wbf_{21} & \Ibf_{n_2} - \Wbf_{22} & -\Wbf_{23} \\ \hdashline
      -\Wbf_{31} & -\Wbf_{32} & \Ibf_{n_3} -
      \Wbf_{33} 
    \end{array}\!\!\right]^{-1} \!\!\!\!= 
    \begin{bmatrix}
      \star & \star & \star \\ \star & \hat \Bbf_1 & \hat \Bbf_2 \\
      \star & \hat \Bbf_3 & \hat \Bbf_4
    \end{bmatrix}
  \end{align*}
  where
  \begin{align*}
    \hat \Bbf_1 &= \Qbf_{22} + \begin{bmatrix} \Qbf_{21} &
      \Qbf_{22} \end{bmatrix} \begin{bmatrix} \Wbf_{13} \\
      \Wbf_{23} \end{bmatrix} \Rbf^{-1} \begin{bmatrix} \Wbf_{31} &
      \Wbf_{32} \end{bmatrix} \begin{bmatrix} \Qbf_{12} \\
      \Qbf_{22} \end{bmatrix},
    \\
    \hat \Bbf_2 &=\begin{bmatrix} \Qbf_{21} &
      \!\!\!\!\Qbf_{22} \end{bmatrix} \!\! \begin{bmatrix} \Wbf_{13} \\
      \Wbf_{23} \end{bmatrix} \! \Rbf^{-1},
    \ \hat \Bbf_3 = \Rbf^{-1} \! \begin{bmatrix} \Wbf_{31} &
      \!\!\!\!\Wbf_{32} \end{bmatrix} \!\! \begin{bmatrix} \Qbf_{12} \\
      \Qbf_{22} \end{bmatrix},
  \end{align*}
  and $\hat \Bbf_4 = \Rbf^{-1}$. Therefore, $-\Gammab$ is the principal pivot transform of
  $\begin{bmatrix} \hat \Bbf_1 & \hat \Bbf_2 \\ \hat \Bbf_3 & \hat
    \Bbf_4 \end{bmatrix}$ so if $\Ibf - \Wbf \in \Pc$,
  Lemma~\ref{lem:p-mat}(v) and the block structure of~\eqref{eq:gamma-stars} guarantee that $-\Mbf_{\sigmab_1} \Mbf_{\sigmab_2}^{-1} \in \Pc$.
\end{proof}

The following result is used in the proof of Theorem~\ref{thm:dec-sub2}.

\begin{lemma}\longthmtitle{GES of cascaded
    interconnections}\label{lem:cas}
  Consider the cascaded dynamics
  \begin{align}\label{eq:w_bk-sim}
    \notag \tau \dot \xbf^\ssm &= -\xbf^\ssm, 
    \\
    \tau \dot \xbf^\ssp &= -\xbf^\ssp + [\Wbf^{\ssp\ssm} \xbf^\ssm +
    \Wbf^{\ssp\ssp} \xbf^\ssp + \tilde \dbf^\ssp]_\zeros^{\mbf\ssp},
  \end{align}
  where $\xbf^\ssm \in \real^r$ and $\xbf^\ssp \in \real^{n - r}$. If
  $\Wbf^{\ssp\ssp}$ is such that
  \begin{align}\label{eq:cas2}
    \tau \dot \xbf^\ssp = -\xbf^\ssp + [\Wbf^{\ssp\ssp} \xbf^\ssp + \tilde
    \dbf^\ssp]_\zeros^{\mbf\ssp},
  \end{align}
  is GES for any constant $\tilde \dbf^\ssp \in \real^{n - r}$, then
  the whole dynamics~\eqref{eq:w_bk-sim} is also GES for any constant
  $\tilde \dbf^\ssp$.
\end{lemma}
\begin{proof}
  We only prove the result for $\tilde \dbf^\ssp = \zeros$. This is
  without loss of generality, since for $\tilde \dbf^\ssp \neq
  \zeros$, we can apply the change of variables $\xib = \xbf -
  \xbf^*$, where $\xbf^*$ is the equilibrium corresponding to input
  $[\zeros^T \ (\tilde \dbf^\ssp)^T]^T$, and shift the equilibrium to
  the origin.  Since~\eqref{eq:cas2} is GES, \cite[Thm
  A.1]{EN-JC:21-tacII} guarantees that there exists $\xbf^\ssp \mapsto
  V^\ssp(\xbf^\ssp)$ such that
  \begin{subequations}
    \begin{align}
      &c_1 \|\xbf^\ssp\|^2 \le V^\ssp(\xbf^\ssp) \le c_2 \|\xbf^\ssp\|^2,
      \\
      &\Big\|\frac{\partial V^\ssp(\xbf^\ssp)}{\partial \xbf^\ssp}\Big\| \le
      c_3 \|\xbf^\ssp\|, \label{eq:conv-Lyap2}
    \end{align}
    for some $c_1, c_2, c_3> 0$, and, if $\xbf^\ssp(t)$ is the solution
    of~\eqref{eq:cas2},
    \begin{align}
      \tau \frac{d}{d t} V^\ssp(\xbf^\ssp(t)) \le -c_4
      \|\xbf^\ssp\|^2, \label{eq:conv-Lyap3}
    \end{align}
  \end{subequations}
  for some $c_4 > 0$. Since $[\cdot]_\zeros^{\mbf\ssp}$ is Lipschitz
  continuous, it follows from~\eqref{eq:conv-Lyap2}
  and~\eqref{eq:conv-Lyap3} that if $\xbf^\ssp(t)$ is the solution
  of~\eqref{eq:w_bk-sim},
  \begin{align*}
    \tau \frac{d}{d t} V^\ssp(\xbf^\ssp(t)) &\le -c_4 \|\xbf^\ssp\|^2 + c_3
    \|\xbf^\ssp\| \|\Wbf^{\ssp\ssm} \xbf^\ssm\|
    \\
    &\le -\frac{c_4}{2} \|\xbf^\ssp\|^2 + \frac{c_3^2 \|\Wbf^{\ssp\ssm}\|^2}{2
      c_4} \|\xbf^\ssm\|^2,
  \end{align*}
  where the second inequality follows from Young's
  inequality~\cite[p. 466]{HKK:02}.  Now, let $V(\xbf) = (c_3^2
  \|\Wbf^{\ssp\ssm}\|^2/2 c_4) \|\xbf^\ssm\|^2 + V^\ssp(\xbf^\ssp)$.
  It is straightforward to verify that $V$ satisfies all the
  assumptions of~\cite[Thm 4.10]{HKK:02} with $a = 2$.
\end{proof}

\vspace*{-2ex}

\begin{IEEEbiography}[{\includegraphics[width=1in, height=1.25in, clip, keepaspectratio]{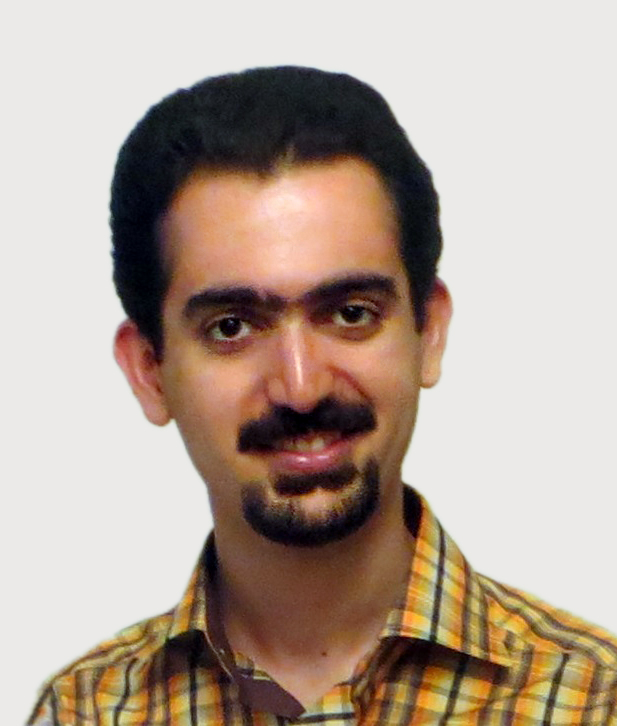}}]{Erfan Nozari}
  received his B.Sc. degree in Electrical Engineering-Control in 2013
  from Isfahan University of Technology, Iran and Ph.D. in Mechanical
  Engineering and Cognitive Science in 2019 from University of
  California San Diego. He is currently a postdoctoral researcher at
  the University of Pennsylvania Department of Electrical and Systems
  Engineering. He has been the (co)recipient of the 2019 IEEE
  Transactions on Control of Network Systems Outstanding Paper Award,
  the Best Student Paper Award from the 57th IEEE Conference on
  Decision and Control, the Best Student Paper Award from the 2018
  American Control Conference, and the Mechanical and Aerospace
  Engineering Distinguished Fellowship Award from the University of
  California San Diego. His research interests include dynamical
  systems and control theory and its applications in computational and
  theoretical neuroscience and complex network systems.
\end{IEEEbiography}

\begin{IEEEbiography}
  [{\includegraphics[width=1in,height=1.25in,clip,keepaspectratio]{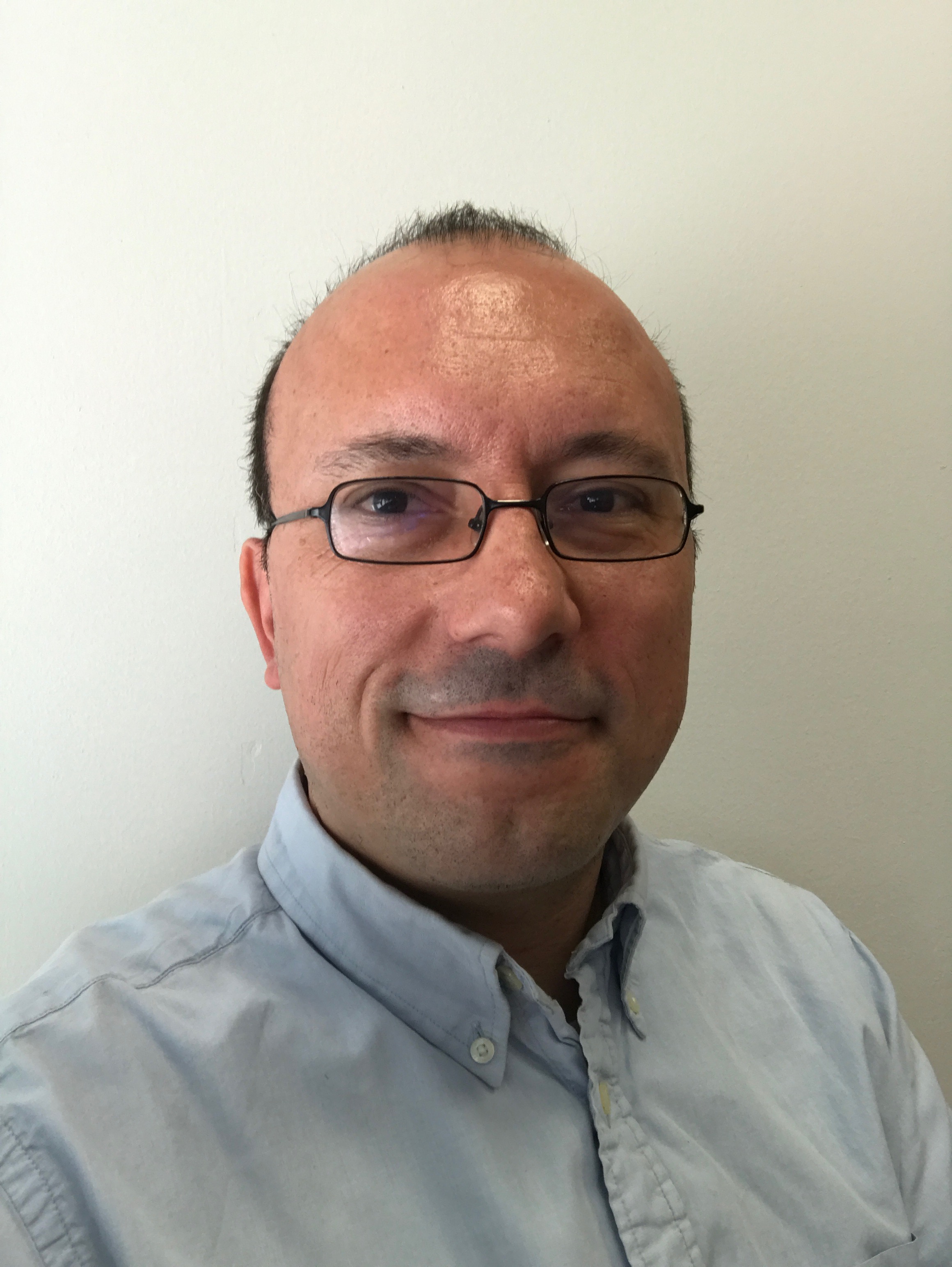}}]
  {Jorge Cort\'es} (M'02-SM'06-F'14) received the Licenciatura degree
  in mathematics from Universidad de Zaragoza, Zaragoza, Spain, in
  1997, and the Ph.D. degree in engineering mathematics from
  Universidad Carlos III de Madrid, Madrid, Spain, in 2001. He held
  postdoctoral positions with the University of Twente, Twente, The
  Netherlands, and the University of Illinois at Urbana-Champaign,
  Urbana, IL, USA. He was an Assistant Professor with the Department
  of Applied Mathematics and Statistics, University of California,
  Santa Cruz, CA, USA, from 2004 to 2007. He is currently a Professor
  in the Department of Mechanical and Aerospace Engineering,
  University of California, San Diego, CA, USA. He is the author of
  Geometric, Control and Numerical Aspects of Nonholonomic Systems
  (Springer-Verlag, 2002) and co-author (together with F. Bullo and
  S. Mart{\'\i}nez) of Distributed Control of Robotic Networks
  (Princeton University Press, 2009).  He is a Fellow of IEEE and
  SIAM. His current research interests include distributed control and
  optimization, network science, opportunistic state-triggered
  control, reasoning and decision making under uncertainty, and
  distributed coordination in power networks, robotics, and
  transportation.
\end{IEEEbiography}

\end{document}